\documentclass[runningheads]{llncs}

\usepackage{graphicx,epic,eepic,amssymb}

\newcommand{\croix}{\mbox{\scriptsize $\times$}}
\newcommand{\compose}{\mbox{\tiny o}}
\newcommand{\entier}{\mbox{$\mathbb{N}$}}
\newcommand{\impli}{\mbox{\small$\Longrightarrow$}}
\newcommand{\arrow}{\mbox{\small$\longrightarrow$}}
\newcommand{\fleche}[1]{\mbox{$\mathop{\arrow}\limits^{#1}$}}
\newcommand{\nofleche}[2]{\mbox{$\ {\tiny/}\hspace*{-1.1em}\mathop{\arrow}
                                           \limits^{#1}_{\mbox{\tiny$#2$}}$}}
\newcommand{\flecheInd}[2]{\mbox{$\mathop{\arrow_{\mbox{\tiny$#2$}}}
\limits^{\!\!\!\!\!#1}$}}
\newcommand{\chemin}[1]{\mbox{$\mathop{\impli}\limits^{#1}$}}
\newcommand{\cheminInd}[2]{\mbox{$\mathop{\Longrightarrow_{\mbox{\tiny$#2$}}}
\limits^{\!\!\!\!\!#1}$}}
\newcommand{\InfSup}[1]{\mbox{$\mathop{<}\!{#1}\!\mathop{>}$}}

\newcommand{\mBAS}[1]{\mbox{\raisebox{-0.75ex}{\scriptsize \rm $\!#1$}}}
\newcommand{\enonce}[2]{\par{\leftskip0em\vspace{1em}\noindent
                             {\bf #1.\ }{\it #2}\par\vspace{1em}}
                        \noindent{\it Proof.}\\[0.0em]}
\newcommand{\cqfd}{\\\mbox{\large\bf $\blacktriangleleft$}\\}

\begin{document}

\title{Synchronizable functions on integers}

\titlerunning{Synchronizable functions}

\author{Didier Caucal \and Chlo\'e Rispal}

\authorrunning{D. Caucal and C. Rispal}

\institute{CNRS, LIGM, University Gustave Eiffel, France\\
\email{\{didier.caucal,chloe.rispal\}@univ-eiffel.fr}}

\maketitle 

\begin{abstract}
{\indent}For all natural numbers \,$a,b$ \,and \,$d > 0$, we consider the 
function \,$f_{a,b,d}$ \,which associates \,$\frac{n}{d}$ \,to any integer 
\,$n$ \,when it is a multiple of \,$d$, and \,$an + b$ \,otherwise; in 
particular \,$f_{3,1,2}$ \,is the Collatz function. 
Coding in base \,$a > 1$ \,with \,$b < a$, we realize these functions by 
input-deterministic letter-to-letter transducers with additional output final
words. 
This particular form allows to explicit, for any integer~\,$n$, the composition 
\,$n$ \,times of such a transducer to compute \,$f^{\,n}_{a,b,d}$. 
We even realize the closure under composition \,$f^{\,*}_{a,b,d}$ \,by an 
infinite input-deterministic letter-to-letter transducer with a regular set of 
initial states and a length recurrent terminal function.
\end{abstract}

\section{Introduction}

Functions on integers have been studied as word functions using an integer 
base. With this coding, some functions on integers can be deterministically 
described by $2$-automata~\cite{RS}. 
In 1966, Ginsburg introduced the sequential transducers, having transitions 
labeled by an input letter and an output word, and computing functions 
deterministically in input~\cite{Gi}. 
Since 1977, Schutzenberger has extended sequential transducers with a terminal 
function associating an output word to each terminal state \cite{Sc}. 

For all natural numbers \,$a,b$ \,and \,$d > 0$, we consider the integer 
function \,$f_{a,b,d}$ \,which 
associates \,$\frac{n}{d}$ \,to any integer~\,$n$ \,when it is a multiple of 
\,$d$ \,and \,$an + b$ \,otherwise. In particular, \,$f_{3,1,2}$ \,is the 
Collatz function~\cite{La}. 
To realize the function \,$f_{a,b,d}$\,, \,the transducer must compute the two 
operations of division by~\,$d$ \,and multiplication by \,$a$. 
The first natural approach is to take the base~\,$d$ \,with the least 
significant digit to the left to see right away if the input is a 
multiple of~\,$d$ \,and, if not, to realize the multiplication by \,$a$ 
\,starting from the left. 
Another way is to choose the base \,$a$ \,to solve the multiplication by \,$a$ 
\,and to realize the division by \,$d$ \,starting from the most significative
digit to the left.

To realize the functions \,$f_{a,b,d}$ \,and its powers, we just need two 
particular forms of left-synchronized and right-synchronized \cite{EM} 
sequential transducers: the prefix and the suffix forms. 
The prefix sequential transducers are the letter-to-letter sequential 
transducers \,{\it i.e.} \,the left-synchronized sequential transducers without 
\,$\varepsilon$-output transition. 
The suffix sequential transducers are the right-synchronized 
sequential transducers without \,$\varepsilon$-input transition. 
Both families of functions realized by prefix and suffix sequential transducers 
are closed under composition, intersection and difference in a quadratic way.

To realize the function \,$f_{a,b,d}$, the choice of the base \,$d$ \,gives 
a simple suffix sequential transducer but it is not suitable to specify, 
for all integers \,$n$, a generic transducer computing its composition \,$n$ 
times. However, the base \,$a$ \,turns out to be more revelant: in the case
where \,$b < a$, which is the case of the Collatz function, it allows to
define a prefix sequential transducer computing \,$f_{a,b,d}$ \,and to
explicit, for any integer \,$n$, the composition \,$n$ \,times of this
transducer. Even better, still under the condition \,$b < a$, we construct
an infinite input-deterministic prefix transducer realizing the
iteration of the composition \,$f^{\,*}_{a,b,d}$. Its terminal function is not
simply defined as the union, for all integers \,$n$ \,of the previous
terminal functions for the \,$f^{\,n}_{a,b,d}$, but we give a length
recurrent definition of it. Finally, we give a geometric representation
of this transducer in the form of a cone where for each integer \,$n$, the
transducer of the division by~\,$d^n$ \,is represented by a circular section 
and the terminal function by transitions going from one section to its 
adjacent smaller section.

\section{Transducers}

{\indent}A (finite) transducer \cite{RS,EM,Be} is a finite 2-automaton: it is 
labeled by pairs of words, or more simply by pairs made with letters or the 
empty word. 
We recall the composition of transducers and its iteration.\\

Let \,$A$ \,be an {\it input alphabet} \,and \,$B$ \,be an {\it output 
alphabet}. A {\it transducer} \,${\cal T} \,= \,(Q,I,F,T)$ \,over 
\,$(A,B)$, or over \,$A$ \,when \,$B = A$, is defined by a set \,$Q$ \,of 
{\it states}, two subsets \,$I$ \,and \,$F$ \,of~\,$Q$ \,of resp.  
{\it initial states} \,and {\it final states}, and a subset \,$T$ \,of 
\,$Q{\croix}(A \cup \{\varepsilon\}){\croix}(B \cup \{\varepsilon\}){\croix}Q$. 
Each \,$(p,a,b,q) \in T$, also denoted by \,$p\ \fleche{a/b}_T\ q$ 
\,or \,$p\ \fleche{a/b}\ q$ \,when \,$T$ \,is clear from the context, is a 
{\it transition} \,from state \,$p$ \,to state \,$q$ \,labeled by \,$(a,b)$, 
or more simply by \,$a/b$\,, \,of {\it input} \,$a$ \,and of {\it output}~$b$. 
We write \,$p\ \fleche{a/\cdot}\ q$ \,and \,$p\ \fleche{\cdot/b}\ q$ \,if there 
exists respectively \,$b$ \,and \,$a$ \,such that \,$p\ \fleche{a/b}\ q$. 
By default, a transducer is finite: it has only a finite number of states. 
We represent~\,${\cal T}$ \,by its graph \,$(Q,T)$ \,with incoming arrows to 
mark its initial states and outgoing arrows from its final states. 
Here is a representation of a transducer \,${\cal T}_0$ \,over \,$\{a,b\}$\,:
\begin{center}
\includegraphics{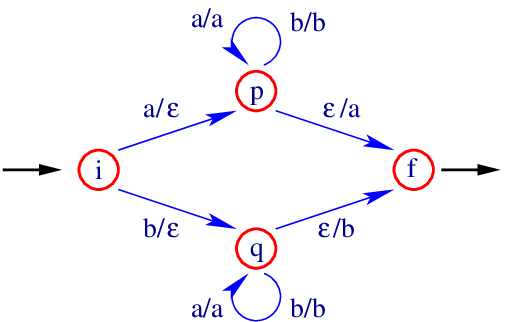}\\
\end{center}
A sequence \,$(p_0,a_1,b_1,p_1\ldots,a_n,b_n,p_n)$ \,of \,$n$ \,consecutive 
transitions
\,$p_0\ \fleche{a_1/b_1}\ p_1$\,,\\
$\ldots$,\,$p_{n-1}\ \fleche{a_n/b_n}\ p_n$ \,is a 
{\it path} \,of {\it length} \,$n$ \,of \,$T$ \,and we 
write \,$p_0\ \cheminInd{u/v}{T}\ p_n$ \,or \,$p_0\ \chemin{u/v}\ p_n$ \,to 
only indicate the {\it source} \,$p_0$\,, the {\it goal} \,$p_n$, the 
{\it input} \,$u = a_1{\ldots}a_n$ \,and the {\it output} 
\,$v = b_1{\ldots}b_n$ \,of the path. The path \,$(p_0)$ \,of length \,$0$ \,is 
denoted by \,$p_0\ \chemin{\varepsilon/\varepsilon}\ p_0$. 
We also write \,$p\ \chemin{u/\cdot}\ q$ \,and \,$p\ \chemin{\cdot/v}\ q$ \,if 
there exists respectively \,$v$ \,and \,$u$ \,such that \,$p\ \chemin{u/v}\ q$. 
A path from an initial vertex to a final vertex is an {\it accepting path}. 
The transducer \,${\cal T}$ \,{\it realizes} \,or \,{\it computes} \,the binary 
relation\\[0.25em]
\hspace*{3em}$\InfSup{\cal T} \ = \ \{\ (u,v) \in A^*{\croix}B^*\ |\ 
\exists\ i \in I\ \,\exists\ f \in F\ (i\,\ \cheminInd{u/v}{T}\ f)\ \}$
\\[0.25em]
of the labels of its accepting paths. The above transducer \,${\cal T}_0$ 
\,realizes the function \,$xu\ \mapsto\ ux$ \,for any word \,$u$ \,and letter 
\,$x$ \,over \,$\{a,b\}$. 
By exchanging the input with the output of the transitions of a transducer 
\,${\cal T} \,= \,(Q,I,F,T)$, we obtain the following 
{\it inverse transducer}\,:\\[0.25em]
\hspace*{3em}${\cal T}^{-1} \,= \,(Q,I,F,T^{-1})$ \ \ where \ \ 
$T^{-1} \,= \,\{\ (p,b,a,q)\ |\ (p,a,b,q) \in T\ \}$\\[0.25em]
which realizes the inverse relation \,$\InfSup{\cal T}^{-1}$ \,of the relation
computed by \,${\cal T}$.\\
We also associate to any transducer \,${\cal T} \,= \,(Q,I,F,T)$ \,its 
{\it mirror transducer}\,:\\[0.25em]
\hspace*{3em}$\widetilde{\cal T} \,= \,(Q,F,I,\widetilde{T})$ \ where \ 
$\widetilde{T} \,= \,\{\ (q,a,b,p)\ |\ (p,a,b,q) \in T\ \}$\\[0.25em]
which reverses the arrows of the transitions without changing the labels and
realizes from \,$F$ \,to \,$I$ \,the mirror of the relation computed by 
\,${\cal T}$\,: 
\,$\InfSup{\widetilde{\cal T}}$ \,is equal to 
\,$\{\,(\tilde{u},\tilde{v})\,|\,(u,v) \in \InfSup{\cal T}\,\}$ 
\,denoted by \,$\widetilde{\InfSup{\cal T}}$.\\
The {\it composition} \,of a transducer \,${\cal T} \,= \,(Q,I,F,T)$ \,over 
\,$(A,B)$ \,with a transducer \,${\cal T}' \,= \,(Q',I',F',T')$ \,over 
\,$(B,C)$ \,is the following transducer over \,$(A,C)$\,:\\[0.25em]
\hspace*{1em}\begin{tabular}{lrcl}
 & ${\cal T}\,\compose\,{\cal T}'$ & \,$=$\, & 
$(Q{\croix}Q',I{\croix}I',F{\croix}F',T\,\compose\,T')$\\[0.25em]
where\, & $T\,\compose\,T'$ & $=$ & 
$\{\ (p,p')\ \fleche{a/\varepsilon}\ (q,p')\ |\ 
p\ \fleche{a/\varepsilon}_T\ q \,\wedge \,p' \in Q'\ \}$\\[0.25em]
 & & $\cup$ & 
$\{\ (p,p')\ \fleche{\varepsilon/b}\ (p,q')\ |\ 
p \in Q \,\wedge \,p'\ \fleche{\varepsilon/b}_{T'}\ q'\ \}$\\[0.25em]
 & & $\cup$ & $\{\ (p,p')\ \fleche{a/c}\ (q,q')\ |\ \exists\ b  
\ (p\ \fleche{a/b}_T\ q \,\wedge \,p'\ \fleche{b/c}_{T'}\ q')\ \}$
\end{tabular}\\[0.5em]
it realizes \,$\InfSup{\cal T}\,\compose\ \InfSup{{\cal T}'}$.\\
The composition \,$n \geq 0$ \,times of a transducer 
\,${\cal T} \,= \,(Q,I,F,T)$ \,over \,$A$ \,is the transducer 
\,${\cal T}^n \,= \,(Q^n,I^n,F^n,T^n)$ \,where 
\,$T^0 \,= \,\{\,\varepsilon\ \fleche{a/a}\ \varepsilon\ |\ a \in A\,\}$ 
\,and\\[0.25em]
\hspace*{1em}\begin{tabular}{rcl}
$T^{n+1}$ & $=$ & $\{\ (p_1,\ldots,p_n,p_{n+1})\ \fleche{a/b}\ 
(q_1,\ldots,q_n,q_{n+1})\ |$\\[0.25em]
 & & \hspace*{0em}$((p_1,\ldots,p_n),p_{n+1})\ \fleche{a/b}_{T^n\,\compose\,T}\ 
((q_1,\ldots,q_n),q_{n+1})\ \}$ \ for any \,$n \geq 0$.
\end{tabular}\\[0.25em]
So \,${\cal T}^n$ \,realizes the relational composition \,$n$ \,times of 
\,$\InfSup{\cal T}$\,: \,$\InfSup{{\cal T}^n} \,= \,\InfSup{\cal T}^n$.\\
In particular \,${\cal T}^0 \,= 
\,(\{\varepsilon\},\{\varepsilon\},\{\varepsilon\},\{\,\varepsilon\ 
\fleche{a/a}\ \varepsilon \mid a \in A\,\}$ \,realizes the identity relation 
on \,$A^*$. 
For any \,$P \subseteq Q$, \,$P^* \,= \,\bigcup_{n \geq 0}P^n$ \,is the set of 
tuples of elements of~\,$P$ \,{\it i.e.} of words over \,$P$.\\
The {\it composition closure} \,of \,${\cal T} \,= \,(Q,I,F,T)$ \,is the 
infinite transducer\\[0.25em]
\hspace*{9em}${\cal T}^* \,= \,(Q^*,I^*,F^*,T^*)$ \ where \ 
$T^* \,= \,\bigcup_{n \geq 0}T^n$\\[0.25em]
which realizes the relation 
\,$\InfSup{{\cal T}^*} \,= \,\bigcup_{n \geq 0}\InfSup{\cal T}^n$.\\
For \,${\cal T}_1 \,= 
\,(\{p\},\{p\},\{p\},\{\,p\ \fleche{a/b}\ p\,,\,p\ \fleche{b/a}\ p\,\})$, its 
composition closure \,${\cal T}_1^*$ \,is\\[0.25em]
\hspace*{0.25em}$(p^*,p^*,p^*,\{p^{2n}\,\fleche{a/a}\,p^{2n},\,p^{2n}\,
\fleche{b/b}\,p^{2n},\,p^{2n+1}\,\fleche{a/b}\,p^{2n+1},\,p^{2n+1}\,\fleche{b/a}
\,p^{2n+1}\,|\,n \geq 0\})$.\\[0.25em]
The relations realized by the (finite) transducers are the 
{\it rational relations} namely those obtained from the finite binary 
relations by applying a finite number of times the {\it rational operations} 
of union, concatenation (componentwise) and its iteration (the Kleene star 
operation).
\begin{proposition}\label{RabinScott}
{\rm\cite{EM}}\,The family of rational relations is closed under composition, 
inverse, mirror, union, concatenation and its iteration.
\end{proposition}
From now on, we consider transducers realizing functions. 
We already notice that the intersection and the difference of two rational 
functions are generally not rational. For instance, the following functions are 
sequential:\\[0.25em]
\hspace*{1em}$f_1 \,= \,\{\ (a^mba^n,a^m)\ |\ m,n \geq 0\ \}$ \ \ and \ \ 
$f_2 \,= \,\{\ (a^mba^n,a^n)\ |\ m,n \geq 0\ \}$\\[0.25em]
but \,$f_1 \cap f_2 \,= \,\{\ (a^nba^n,a^n)\ |\ n \geq 0\ \}$ 
\,and \,$f_1 - f_2 \,= \,\{\ (a^mba^n,a^m)\ |\ m \neq n\ \}$ \,are not rational 
since their domains are not regular languages.

\section{Sequential transducers}

{\indent}The sequential transducers have been defined \cite{Gi,Sc} to 
compute functions in a deterministic way according to the inputs.
The functions are realized from left to right reading one letter at a time on 
input and writing a word on output for each letter.\\

We use transitions of the form \,$\fleche{a/v}$ \,where \,$a$ \,is a letter 
and \,$v$ \,is a word. 
A transducer over \,$(A,B)$ \,can be defined in a {\it terminal form} 
\,${\cal T} \,= \,(Q,I,\omega,T)$ \,with a finite transition set 
\,$T \,\subset \,Q{\croix}A{\croix}B^*{\croix}Q$ \,and a 
{\it terminal function} \,$\omega$ \,from~\,$Q$ \,into the set \,Reg$(B^*)$ 
\,of regular languages over \,$B$. It realizes the relation\\[0.25em]
\hspace*{3em}$\InfSup{\cal T} \ = \ \{\ (u,vw)\ |\ 
\exists\ i \in I\ \,\exists\ f \in {\rm dom}(\omega)\ 
(i\,\ \cheminInd{u/v}{T}\ f \,\wedge \,w \in \omega(f))\ \}$.\\[0.25em]
The domain of \,$\omega$ \,is the set of final states. 
We also represent \,${\cal T}$ \,by its graph \,$(Q,T)$ \,and for any 
\,$q \in {\rm dom}(\omega)$, we add an outgoing arrow from \,$q$ \,labeled 
by \,$\omega(q)$ \,and this label can be omitted when 
\,$\omega(q) = \{\varepsilon\}$. 
Here is a terminal form representation of \,\,${\cal T}_0$\,:
\begin{center}
\includegraphics{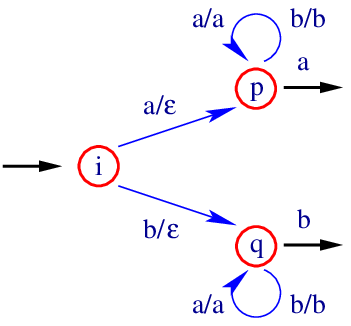}\\
\end{center}
Note that any transducer \,${\cal T} \,= \,(Q,I,F,T)$ \,has a terminal form 
\,${\cal T}' \,= \,(Q,I,\omega',T')$ \,such that 
\,$\InfSup{{\cal T}'} \,= \,\InfSup{\cal T}$\,; it is defined by\\[0.25em]
\hspace*{6em}\begin{tabular}{rcl}
$T'$ & $=$ & $\{\ p\ \fleche{a/vb}\ q\ |\ a \in A \,\wedge \,\exists\ r 
\,(p\ \chemin{\varepsilon/v}_T\ r\ \fleche{a/b}_T\ q)\ \}$\\[0.25em]
$\omega'(q)$ & $=$ & $\{\ v\ |\ \exists\ f \in F \ 
(q\ \chemin{\varepsilon/v}_T\ f)\ \}$.
\end{tabular}\\[0.25em]
The terminal form is well suited to define the determinism on input.\\
We say that a transition set \,$T \,\subset \,Q{\croix}A{\croix}B^*{\croix}Q$ 
\,is {\it input-deterministic} \,if\\[0.25em]
\hspace*{9em}$(p\ \fleche{a/u}_T\ q \,\wedge \,p\ \fleche{a/v}_T\ r) \ \ 
\Longrightarrow \ \ (u = v \,\wedge \,q = r)$.\\[0.25em]
We also say that \,$T$ \,is {\it input-complete} \,if \ 
$q\ \fleche{a/\cdot}$ \,for any \,$q \in Q$ \,and \,$a \in A$.\\
Such a general approach allows us to compute functions (resp. mappings) in a
deterministic (and complete) way \cite{Sc}.\\
A {\it sequential transducer} \,over \,$(A,B)$ \,is a transducer of terminal 
form \,$(Q,i,\omega,T)$ \,with a unique initial state \,$i$, a terminal 
function \,$\omega : Q\ \fleche{}\ B^*$ \,and an input-deterministic transition 
set \,$T \,\subset \,Q{\croix}A{\croix}B^*{\croix}Q$. 
In particular \,${\cal T}_0$ \,is a sequential transducer. 
The relation realized by a sequential transducer is a 
{\it sequential function}. 
The composition of sequential transducers \,${\cal T} \,= \,(Q,I,\omega,T)$ 
\,by \,${\cal T}' \,= \,(Q',I',\omega',T')$ \,is then expressed by
\\[0.25em]
\hspace*{1em}\begin{tabular}{lrcl}
 & ${\cal T}\,\compose\,{\cal T}'$ & $=$ & 
$(Q{\croix}Q',I{\croix}I',\omega\,\compose\,\omega',T\,\compose\,T')$\\[0.25em]
where & $T\,\compose\,T'$ & $=$ & 
$\{\ (p,p')\ \fleche{a/v}\ (q,q')\ |\ \exists\ u 
\ (p\ \fleche{a/u}_T\ q \,\wedge \,p'\ \chemin{u/v}_{T'}\ q')\ \}$\\[0.25em]
and & $\omega\,\compose\,\omega'((q,p'))$ & $=$ & 
$v.\omega'(q')$ \,for any \,$q \in {\rm dom}(\omega)$, 
$q' \in {\rm dom}(\omega')$, $p'\ \cheminInd{\omega(q)/v}{T'}\ q'$.
\end{tabular}\\
\begin{proposition}\label{Schutzenberger}
{\rm\cite{Sc}}\,The sequential functions are preserved by composition.
\end{proposition}
Note that the previous functions \,$f_1$ \,and \,$f_2$ \,are sequential 
whereas \,$f_1 \cap f_2$ \,and \,$f_1 - f_2$ \,are not rational.

\section{Prefix and suffix sequential transducers}

An easy way to realize functions is to use input-deterministic letter-to-letter 
transducers. In the following, we consider two particular forms of transducers
composed by letter-to-letter sequential transducers with additional 
\,$\varepsilon$-input and \,$\varepsilon$-output transitions on the left or on 
the right.

We first recall letter-to-letter transducers. 
A {\it synchronous transducer} \,is a transducer \,${\cal T} \,= \,(Q,I,F,T)$ 
\,over \,$(A,B)$ \,whose transition set \,$T$ \,is {\it letter-to-letter} 
\,{\it i.e.} \,$T \,\subseteq \,Q{\croix}A{\croix}B{\croix}Q$\,; \,it realizes 
a {\it synchonous relation} \,which is {\it length-preserving}\,: \,for any 
\,$(u,v) \in \InfSup{\cal T}$, $|u| = |v|$. 
A transducer which is both sequential and synchronous is a 
{\it synchronous sequential transducer} \,(or a Mealy machine \cite{Me}) 
\,{\it i.e.} \,of the form \,$(Q,i,F,T)$ \,where 
\,$T$ \,is letter-to-letter and input-deterministic. 
Note that the function \,$\{\ (aa,aa)\,,\,(ab,ba)\}$ \,is synchronous and 
sequential but it cannot be realized by a synchronous sequential transducer.\\
A {\it prefix sequential transducer} \,is an input-deterministic 
letter-to-letter transducer followed by \,$\varepsilon$-input transitions via 
the terminal function \,{\it i.e.} \ a sequential transducer 
\,${\cal T} \,= \,(A,i,\omega,T)$ \,where \,$T$ \,is letter-to-letter. 
It realizes a {\it prefix sequential function} \,which is 
{\it length-increasing}\,: \,$|u| \leq |v|$ \,for any 
\,$(u,v) \in \InfSup{\cal T}$.\\
A {\it suffix sequential transducer} \,is an input-deterministic 
letter-to-letter transducer preceded by \,$\varepsilon$-output transitions 
\,{\it i.e.} \ a sequential transducer \,${\cal T} \,= \,(A,i,F,T)$ \,having a 
terminal function reduced to a final state set \,$F$ \,and a transition set 
\,$T \,\subseteq \,Q{\croix}A{\croix}(B \cup \{\varepsilon\}){\croix}Q$ 
\,which is input-deterministic and of {\it initial} 
\,$\varepsilon$-{\it output} \,meaning that 
\,$\varepsilon$-{\it output transitions} \,$\fleche{\cdot/\varepsilon}$ \,can 
only be at the beginning of paths: 
\ $\fleche{\cdot/b}_T\,\fleche{\cdot/\varepsilon}_T \ \ \Longrightarrow \ \ 
b = \varepsilon$. It realizes a {\it suffix sequential function} \,which is 
{\it length-decreasing}\,: \,$|u| \geq |v|$ \,for any 
\,$(u,v) \in \InfSup{\cal T}$. Here is an illustration of these two particular 
forms of synchronized sequential transducers.
\begin{center}
\includegraphics{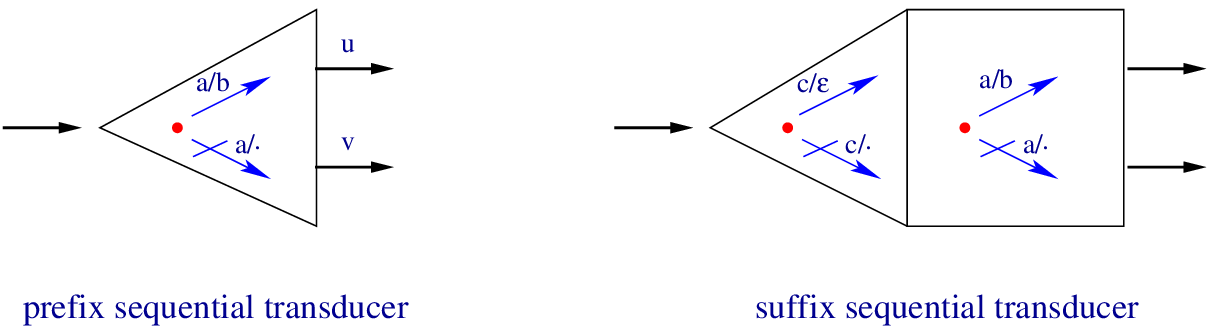}\\
\end{center}
The composition of sequential transducers preserves the prefixity and the 
suffixity.
\begin{proposition}\label{PrefixSuffixSequential}
The prefix and the suffix sequential functions are preserved in quadratic time 
and space under composition, intersection and difference.
\end{proposition}
Encoding integers by words, we use these synchronized sequential transducers 
to describe functions on integers.

\section{Automaticity of functions on integers}

{\indent}The purpose of automaticity is to describe objects by deterministic 
finite automata. Taking an integer base to code integers by words, we consider
the automaticity of functions on integers using sequential transducers. 
Then we refine the automaticity using prefix or suffix sequential transducers.\\

Let an integer \,$a > 1$ \,and \,$\downarrow_a \,= \,\{0,\ldots,a-1\}$ \,be 
the alphabet \,of its {\it digits}.\\
Any word \,$u \in \,\downarrow_a^*$ \,is a {\it representation in base} 
\,$a$ \,of the integer \,$[u]\mBAS{\,a}$ \,defined by:\\[0.25em]
\hspace*{3em}$[c_n{\ldots}c_0]\mBAS{\,a} \ = \ \sum_{i = 0}^nc_ia^i$ \ for any 
\,$n \geq 0$ \,and \,$c_0,\ldots,c_n \in \,\downarrow_a$\\[0.25em]
where the least significant digit is to the right. 
Any word \,$u \in \,\downarrow_a^*$ \,is also a 
{\it reverse representation in base} \,$a$ \,of the integer \,$\mBAS{a}[u]$ 
\,defined by:\\[0.25em]
\hspace*{3em}$\mBAS{a}[c_0{\ldots}c_n] \ = \ \sum_{i = 0}^nc_ia^i$ \ for any 
\,$n \geq 0$ \,and \,$c_0,\ldots,c_n \in \,\downarrow_a$\\[0.25em]
where the least significant digit is to the left. 
Thus \,$\mBAS{a}[u] \,= \,[\widetilde{u}]\mBAS{\,a}$ \,for any 
\,$u \in \,\downarrow_a^*$\,.\\
Representations of integers are extended to relations.\\
A relation \,$R \,\subseteq \,\downarrow_a^*{\croix}\downarrow_a^*$ 
\,is a representation (resp. reverse representation) in base~\,$a$ \,of the 
binary relation \,$[R]\mBAS{\,a}$ \,(resp. $\!\mBAS{a}[R]$) on \,$\entier$\,:
\\[0.25em]
\hspace*{0.5em}$[R]\mBAS{\,a} \ = \ \{\ ([u]\mBAS{\,a\,},[v]\mBAS{\,a})\ |\ 
(u,v) \in R\ \}$ \ \ and \ \ 
$\mBAS{a}[R] \ = \ \{\ (\mBAS{\,a}[u]\,,\mBAS{\,a}[v])\ |\ (u,v) \in R\ \}$.
\\[0.25em]
For \,$\widetilde{R} \,= \,\{\ (\widetilde{u},\widetilde{v})\ |\ (u,v) \in R\ 
\}$ \,the mirror of \,$R$, we have 
\,\,$\mBAS{a}[R] \,= \,[\widetilde{R}]\mBAS{\,a}$\,.\\
Then, functions on integers can be realized by sequential transducers.
\begin{definition}\label{Automatic}
A function \,$f : \entier\ \fleche{}\ \entier$ \,is {\it automatic} (resp. 
{\it reverse automatic}) if there exists an integer base \,$a > 1$ \,and 
a sequential transducer \,${\cal T}$ \,over \,$\downarrow_a$ \,such that 
\,$f \,= \,[{\rm R}({\cal T})]\mBAS{\,a}$ 
(resp. \,$f \,= \mBAS{\,a}[{\rm R}({\cal T})]$\,).
\end{definition}
We refine the automaticity by using prefix or suffix sequential transducers. 
We say that a function \,$f : \entier\ \fleche{}\ \entier$ \,is 
{\it prefix/suffix automatic} (resp. {\it reverse prefix/suffix automatic}) if 
there exists an integer base \,$a > 1$ \,and a prefix/suffix sequential 
transducer\,~${\cal T}$ \,over \,$\downarrow_a$ \,such that 
\,$f \,= \,[{\rm R}({\cal T})]\mBAS{\,a}$ 
(resp. \,$f \,= \mBAS{\,a}[{\rm R}({\cal T})]$).\\
In the following, we study the prefix and suffix automaticity of the mappings 
\,$f_{a,b,d} : \entier\ \fleche{}\ \entier$ \,for all natural numbers 
\,$a,b,d$ \,with \,$d \neq 0$ \,defined for any integer \,$n \geq 0$ \,by
\\[0.5em]
\hspace*{6em}{$f_{a,b,d}(n)\ =\ 
\left\{\begin{tabular}{ll}
$\frac{n}{d}$ & \ if \ \ $n$ \,is a multiple of \,$d$,\\[0.25em]
$an + b$ & \ otherwise.
\end{tabular}\right.$}\\[0.25em]
Note that \,$f_{3,1,2}$ \,is the Collatz function.

\section{Suffix automaticity of functions \,${\bf f_{a,b,d}}$}

{\indent}To realize the functions \,$f_{a,b,d}$ \,with sequential transducers, 
a natural and simple way is to take the base \,$d$ \,to perform by shift the 
division by \,$d$. Taking the least significant digit to the left allows to 
test the multiplicity by \,$d$ \,and, if not, to perform the multiplication by 
\,$a$ \,and the addition of \,$b$. 
Simple suffix transducers will be sufficient to describe these operations. 
We therefore recall the realization of the multiplication by \,$a$ \,by a 
synchronous sequential transducer.\\

The function \,$f_{a,b,1}$ \,is the identity mapping which can be realized in 
any base by the trivial synchronous sequential transducer with only one state 
and the labels of the loops are all pairs of identical digits. 
In this section, we assume that \,$d > 1$. 
For any natural number \,$a$, we realize the function \,$n\ \mapsto\ an$ \,in 
base~\,$d$ \,with the least significant digit to the left by the synchronous 
sequential transducer 
\,$_{d,a}* \,= \,(\downarrow_{a\,},\{0\},\{0\},_{\,d,a}\!\croix)$ \,where\\[0.25em]
\hspace*{3em}$i\ \fleche{b/c}_{_{d,a}\croix}\ j$ \ \ if  \ 
$ab + i \,= \,c + d\,j$ \ for any \,$i,j \in \,\downarrow_{a}$ \,and 
\,$b,c \in \,\downarrow_{d}$\,.\\
Here is an illustration of the transducer \,$_{2,4}*$ \,realizing a reverse 
representation in base $2$ of the multiplication by \,$4$\,:
\begin{center}
\includegraphics{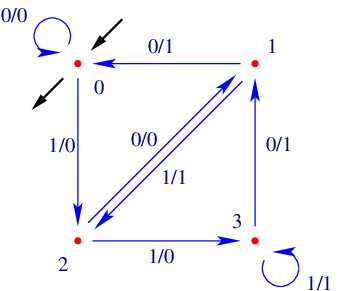}\\
\end{center}
The multiplication defined on digits by transitions extends to the 
multiplication of numbers by paths. This property holds even for transitions 
\,$i\ \fleche{b/c}_{_{d,a}\croix}\ j$ \,where \,$i$ \,and \,$j$ \,are not 
necessarily lower than the multiplicator \,$a$.
\begin{lemma}\label{PathMult}
For any \,$i,j \geq 0$ \,and \,$u,v \in \,\downarrow_d^*$\,, \,we have\\[0.25em]
\hspace*{6em}$i\,\ \chemin{u/v}_{_{d,a}\croix}\ j \ \ \ \Longleftrightarrow \ \ \ 
\mBAS{d}[u]a + i \,= \,\mBAS{\,d}[v] + j\,d^{|u|}$ \,and \ $|u| = |v|$.
\end{lemma}
Applying Lemma~\ref{PathMult}, the synchronous sequential transducer 
\,$_{d,a}*$ \,realizes the multiplication by \,$a$ \,in base \,$d$\,:\\[0.25em]
\hspace*{3em}$\InfSup{_{d,a}*} \,= 
\,\{\ (u,v)\ |\ u,v \in \,\downarrow_d^* \ \wedge \ |u| = |v| \ \wedge 
\ \mBAS{d}[v] \,= \,\mBAS{d}[u]a\ \}$\\[0.25em]
which is a reverse representation in base \,$d$ \,of the function 
\,$n\ \mapsto\ an$.\\
Also by Lemma~\ref{PathMult} and for any \,$b \geq 0$, the synchronous 
sequential transducer\\
\hspace*{3em}$_{d,a,b}* \ = \ (\downarrow_{a,b},\{b\},\{0\},_{\,d,a,b}\!\croix)$ 
\,where\\[0.25em]
\hspace*{3em}$_{\,d,a,b}\croix$ \,is the restriction of 
\,$_{\,d,a}\croix$ \,to \,$\downarrow_{a,b} \,= \,\{0,\ldots,{\rm max}(a-1,b)\}$
\\[0.25em]
realizes the function 
\,$\{\ (u,v)\ |\ u,v \in \,\downarrow_d^* \ \wedge \ |u| = |v| \ \wedge 
\ \mBAS{d}[v] \,= \,\mBAS{d}[u]a + b\ \}$ \,which multiplies by \,$a$ \,and 
adds \,$b$.\\
We give a representation of the transducer \,$_{2,3,7}*$ \,realizing the 
multiplication by~\,$3$ \,and addition of \,$7$ \,in base \,$2$\,:
\begin{center}
\includegraphics{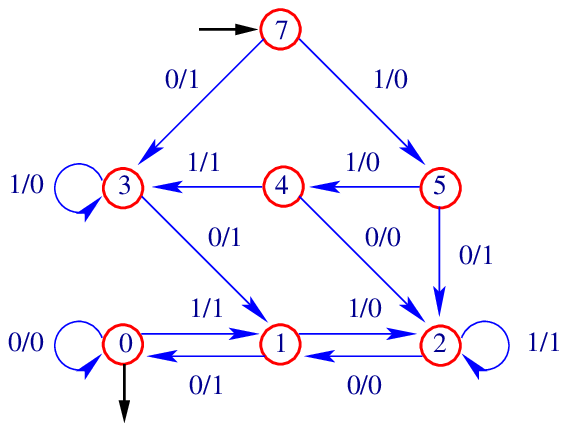}\\
\end{center}
without the state \,$6$ \,which is not accessible from the initial state \,$7$.
\\
We can now present a suffix sequential transducer realizing a reverse 
representation in base~\,$d$ \,of the function \,$f_{a,b,d}$\,.\\
First of all, the function \,$\{\ (dn,n)\ |\ n \in \entier\ \}$ \,has for 
reverse representation in base \,$d$ \,the function 
\,$\{\ (0u,u)\ |\ u \in \,\downarrow_d^*\ \}$ \,which is realized by the 
suffix sequential transducer\\[0.25em]
\hspace*{6em}${\cal D} \,= \,(\{\alpha,\beta\},\{\alpha\},\{\beta\},
\{\alpha\ \fleche{0/\varepsilon}\ \beta\} \cup 
\{\beta\ \fleche{c/c}\ \beta \mid c \in \,\downarrow_d\,\})$.\\[0.25em]
We now give a synchronous sequential transducer to compute a reverse 
representation in base \,$d$ \,of \,$g_{a,b,d} \,= 
\,\{\ (dn+c\,,\,a(dn+c)+b)\ |\ c,n \in \entier \,\wedge \,0 < c < d\ \}$.\\
This is realized by the previous transducer \,$_{d,a,b}*$ \,except that an 
initial transition cannot be of input \,$0$ \,{\it i.e.} \,of the form 
\,$b\ \fleche{0/\cdot}$\,. 
We just have to take a new initial state \,$\alpha$ \,and the following 
synchronous sequential transducer:\\[0.25em]
\hspace*{3em}\begin{tabular}{lrcl}
 & $_{d,a,b}{\cal M}$ & $=$ & 
$(\{\alpha\} \,\cup \downarrow_{a,b}\,,\,\{\alpha\}\,,\,\{0\}\,,\,_{d,a,b}T\,)$
\\[0.25em]
where & $_{d,a,b}T$ & $=$ & 
$\,_{d,a,b}\croix \,\cup \,\{\,\alpha\ \fleche{c/e}\ j \mid c \neq 0 \,\wedge 
\,b\ \fleche{c/e}_{_{\,d,a,b}\!\croix}\ j\,\}$
\end{tabular}\\[0.25em]
which realizes a reverse representation in base \,$d$ \,of \,$g_{a,b,d}$\,. 
Finally, the suffix sequential transducer 
\,${\cal D} \,\cup \,_{d,a,b}{\cal M}$ \,defined by\\[0.25em]
\hspace*{2em}${\large (}\{\alpha,\beta\} \,\cup 
\downarrow_{a,b}\,,\,\{\alpha\}\,,\,\{0,\beta\}\,,\,
_{d,a,b}T \,\cup \{\alpha\ \fleche{0/\varepsilon}\ \beta\} \cup 
\{\beta\ \fleche{c/c}\ \beta \mid c \in \,\downarrow_d\,\}{\large )}$\\[0.25em]
realizes a reverse representation in base \,$d$ \,of \,$f_{a,b,d}$\,.
\begin{proposition}\label{SuffixFabd}
For all \,$a,b \geq 0$ \,and \,$d > 0$, \,$f_{a,b,d}$ \,is reverse suffix 
automatic.
\end{proposition}
Here is an illustration of \,${\cal D} \,\cup \,_{3,1,2}{\cal M}$ \,realizing 
a reverse representation in base~$2$ of the Collatz function:
\begin{center}
\includegraphics{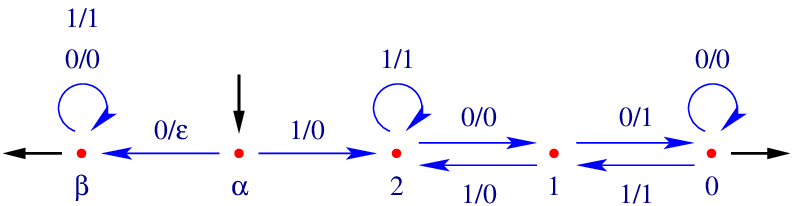}\\
\end{center}
Although this suffix sequential transducer is simple and easy to obtain from 
the definition of the Collatz function, it is tedious to compose it several 
times in order to realize the first powers of the Collatz function. 
Moreover, we are even less able to express in terms of \,$n$ \,the composition 
\,$n$ \,times of this transducer. 
We will see that this becomes possible for any function \,$f_{a,b,d}$ \,using 
prefix sequential transducers but under the restriction that \,$b < a$ 
\,which is a constraint satisfied by the Collatz function.

\section{Prefix automaticity of functions \,${\bf f_{a,b,d}}$}

{\indent}Another way to realize the functions \,$f_{a,b,d}$ \,with sequential 
transducers is to take the base~\,$a$ \,and the least significant digit to the 
right to compute deterministically the division by \,$d$. With this approach 
and for \,$b < a$, we realize \,$f_{a,b,d}$ \,by a prefix sequential 
transducer with \,$d$ \,states.\\

To realize the division by a synchronous sequential transducer, a nice 
prologue has been given in \cite{Sa}. 
Let natural numbers \,$a > 1$, $d > 0$ \,and \,$0 \leq r < d$. 
We realize the function \,$dn + r\ \mapsto\ n$ \,in base \,$a$ \,by just taking 
the inverse and mirror of the previous transducer \,$_{a,d,r}*$ \,for the 
multiplication. We define the synchronous sequential transducer 
\,$/_{a,d,r} \,= \,(\downarrow_{d\,},\{0\},\{r\},:_{a,d})$ \,of the division by 
\,$d$ \,in base \,$a$ \,with remainder \,$r$ \,where\\[0.25em]
\hspace*{3em}$i\ \fleche{b/c}_{:_{a,d}}\ j$ \ \ if \ $ia + b \,= \,cd + j$ \ \ 
for all \,$i,j \in \ \downarrow_{d}$ \,and \,$b,c \in \ \downarrow_{a}$.
\\[0.25em]
Here is a representation of the transducers \,$/_{2,3,2}$ \,and \,$/_{3,2,0}$\,:
\begin{center}
\includegraphics{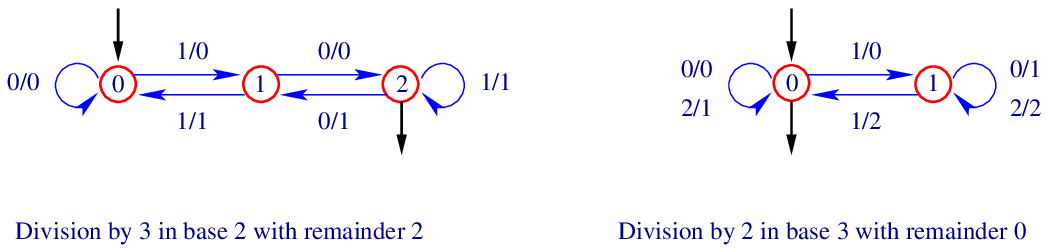}\\
\end{center}
Let us apply Lemma~\ref{PathMult}.
\begin{lemma}\label{CheminDiv}
For all \,$i,j \in \,\downarrow_d$ \,and \,$u,v \in \,\downarrow_a^*$\,, \,we 
have\\[0.25em]
\hspace*{6em}$i\,\ \chemin{u/v}_{:_{a,d}}\ j \ \ \ \Longleftrightarrow \ \ \ 
i\,a^{|u|} + [u]\mBAS{\,a} \,= \,[v]\mBAS{\,a}d + j$ \,and \ $|u| = |v|$.
\end{lemma}
Thus \,$/_{a,d,r}$ \,realizes 
\,$\{\ (u,v)\ |\ u,v \in \,\downarrow_a^* \ \wedge \ |u| = |v| \ \wedge 
\ [u]\mBAS{\,a} \,= \,[v]\mBAS{\,a}d + r\ \}$.
\\[0.25em]
Let us propose a way to visualize these transducers to highlight basic 
symmetries. 
The \,$d$ \,integers of the vertex set \,$\{0,\ldots,d-1\}$ \,of \,$:_{a,d}$
\,are equidistant on a counterclockwise circle in a way that the diameter 
between \,$0$ \,and \,$d-1$ \,is horizontal with \,$0$ \,at the top right. 
Here is a representation for respectively \,$d \,= \,1,2,3,4$\,:
\begin{center}
\includegraphics{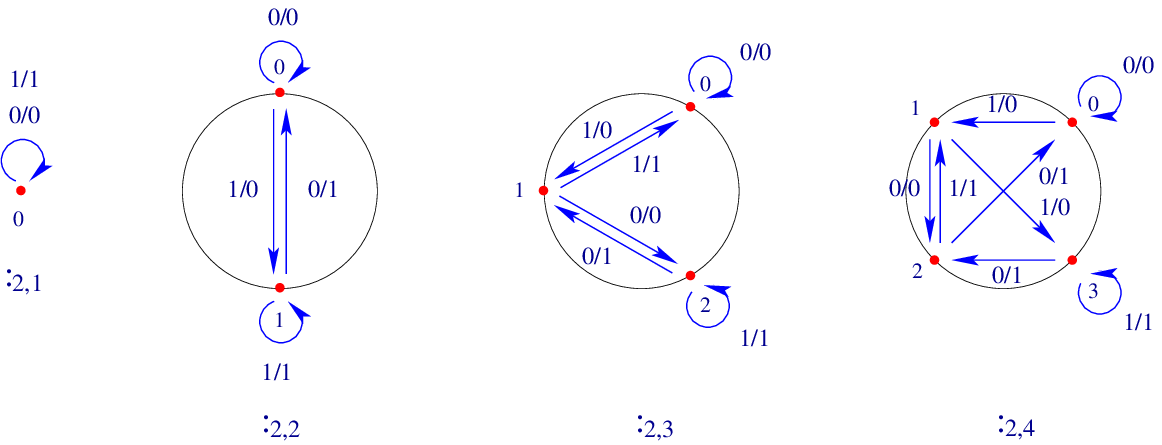}\\
\end{center}
Here is a simple and natural algorithm to draw one by one the transitions of 
the transducer \,$:_{d,a}$ \,realizing the division by \,$d$ \,in base \,$a$.\\
From the Euclidean division of \,$n \,= \,dq + r$, the division of \,$n+1$ \,is
\\[0.25em]
\hspace*{1em}$n+1 \,= \,dq + (r+1)$ \,if \,$r < d-1$ \ \ and \ \ 
$n+1 \,= \,d(q+1)$ \,if \,$r = d-1$.\\
This allows us to incrementally determine \,$:_{a,d}$ \,as follows:\\[0.25em]
\hspace*{5em}$E \,= \,\emptyset$\\
\hspace*{5em}$output \,= \,0$\\
\hspace*{5em}$goal \,= \,0$\\
\hspace*{5em}For all \,$source$ \,from \,$0$ \,to \,$d-1$ \,Do\\
\hspace*{8em}For all \,$input$ \,from \,$0$ \,to \,$a-1$ \,Do\\
\hspace*{11em}add to \,$E$ \,the transition 
\,$source\ \fleche{input/output}\ goal$\\
\hspace*{11em}add \,$1$ \,to \,$goal$\\
\hspace*{11em}If \,$goal = d$ \,Then\\
\hspace*{14em}add \,$1$ \,to \,$output$\\
\hspace*{14em}$goal = 0$\\
\hspace*{11em}EndIf\\
\hspace*{8em}EndFor\\
\hspace*{5em}EndFor\\
\hspace*{5em}return \,$E$\\[0.25em]
This algorithm processes the \,$d$ \,vertices in increasing order.
For each vertex \,$i$, it determines its \,$a$ \,transitions by increasing 
input \,$b$. 
The goal \,$j$ \,is the number \,$ia + b$ \,of transitions already given  
modulo \,$d$, and the output \,$c$ \,is its quotient by \,$d$, 
which is justified by the equality \,$ia + b \,= \,cd + j$. 
Here is an illustration of this construction process.
\begin{center}
\includegraphics{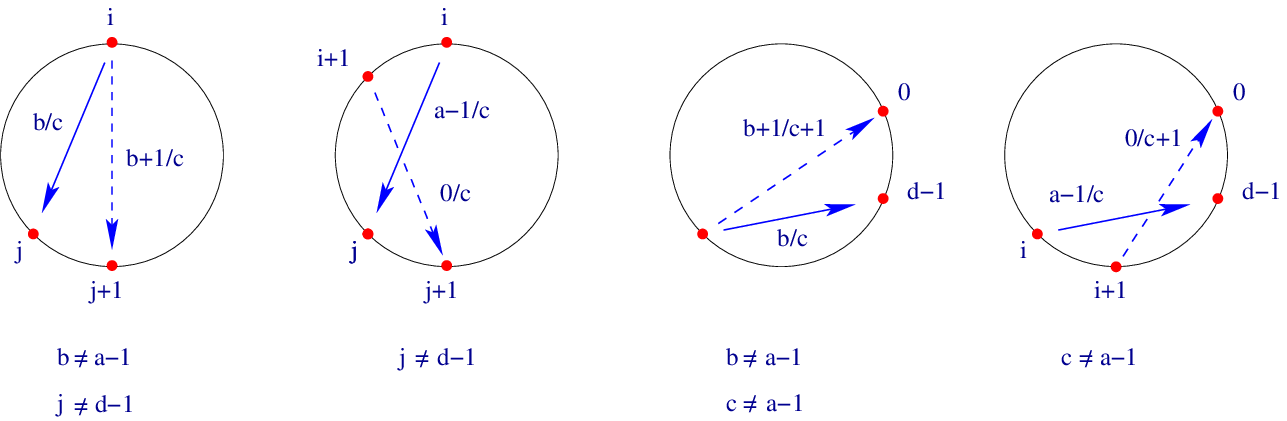}\\
\end{center}
We perform a single turn for the sources and \,$a$ \,turns for the goals. 
We begin with the loop \,$0\ \fleche{0/0}\ 0$ \,and we end with the loop 
\,$d-1\ \fleche{a-1/a-1}\ d-1$.\\
It follows the illustration below of the Euclidean division \,$:_{5,8}$ \,by 
\,$8$ \,in base \,$5$\,:
\begin{center}
\includegraphics{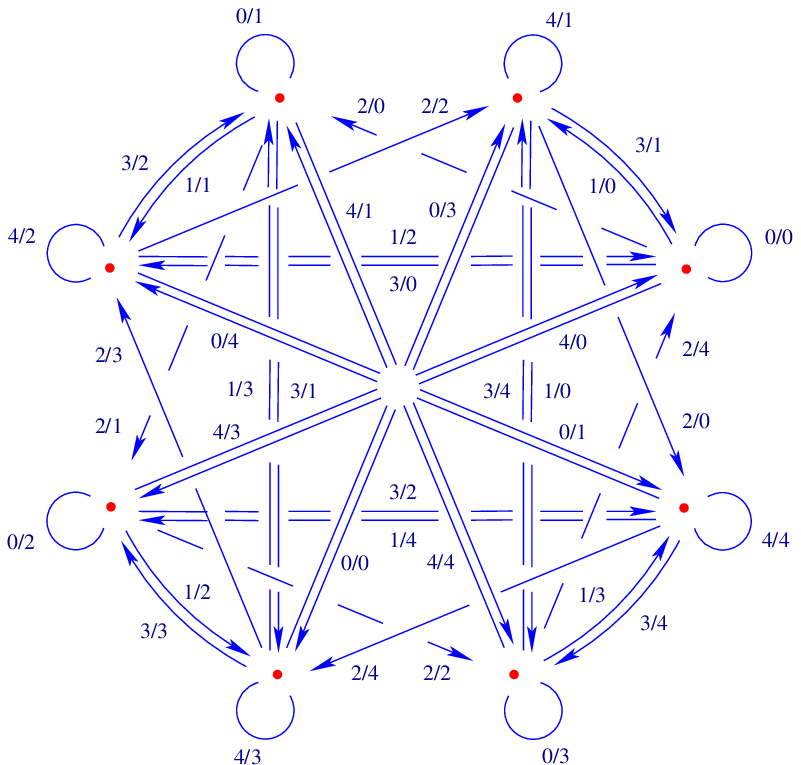}\\
\end{center}
The disposition of the vertices is useful to identify, among others, properties 
of symmetry for these transducers.\\
We have to give a prefix transducer realizing a representation in base \,$a$ 
\,of the function \,$f_{a,b,d}$ \,for \,$b < a$. 
We start with \,$d = 2$ \,namely the functions \,$f_{a,b}$ \,defined for any
integer \,$n \geq 0$ \,by\\[0.5em]
\hspace*{9em}{$f_{a,b}(n)\ =\ 
\left\{\begin{tabular}{ll}
$\frac{n}{2}$ & \ if \ \ $n$ \,is even,\\[0.25em]
$an + b$ & \ otherwise.
\end{tabular}\right.$}\\[0.25em]
In the case where \,$a$ \,and \,$b$ \,are of the same parity, we also consider 
the function\\[0.5em]
\hspace*{9em}{$f'_{a,b}(n)\ =\ 
\left\{\begin{tabular}{ll}
$\frac{n}{2}$ & \ if \ \ $n$ \,is even,\\[0.25em]
$\frac{an + b}{2}$ & \ otherwise
\end{tabular}\right.$}\\[0.25em]
which is an {\it acceleration} \,of \,$f_{a,b}$\,. 
A transducer realizing \,$f'_{a,b}$ \,can be obtained from the transducer 
\,$/_{a,2,0}$ \,of the division by \,$2$ \,in base \,$a$ \,with an additional 
terminal function. If an initial path ends to the state \,$0$, the input is a 
representation in base \,$a$ \,of an even integer \,$n$ \,and the output 
represents \,$\frac{n}{2} \,= \,f'_{a,b}(n)$. 
If an initial path ends to the state \,$1$, the input represents an odd 
integer \,$n$ \,and the output represents \,$\frac{n-1}{2}$. So the digit of 
the terminal function must be \,$\frac{a+b}{2}$ \,to get 
\,$a\,\frac{n-1}{2} \,+ \,\frac{a+b}{2} \,= \,\frac{an+b}{2} \,= \,f'_{a,b}(n)$.
\begin{proposition}\label{Transducf'}
For all \,$0 \leq b < a$ \,with \,$a > 1$ \,and \,$a,b$ \,of the same 
parity,\\[0.25em]
\hspace*{1em}${\cal G}'_{a,b} \,= \,(\{0,1\},\{0\},\omega'_{a,b}\,,:_{a,2})$ \ 
with \ 
$\omega'_{a,b}(0) = \varepsilon$ \,and \,$\omega'_{a,b}(1) = \frac{a+b}{2}$
\\[0.25em]
is a prefix sequential transducer realizing a representation in base \,$a$ \,of 
\,$f'_{a,b}$\,.
\end{proposition}
\begin{proof}
As \,$:_{a,2}$ \,is input-deterministic and input-complete, for all 
\,$u \in \,\downarrow_a^*$\,, there exists a unique \,$v \in \,\downarrow_a^*$ 
\,and \,$j \in \{0,1\}$ \,such that \,$0\ \chemin{u/v}_{:_{a,2}}\ j$.\\
By Lemma~\ref{CheminDiv}, we have 
\,$[u]\mBAS{\,a} \,= \,2\,[v]\mBAS{\,a} + j$.\\
For \,$j = 0$, \,$[u]\mBAS{\,a}$ \,is even and 
\,$[v]\mBAS{\,a} \,= \,\frac{[u]\mBAS{\,a}}{2} \,= \,f'_{a,b}([u]\mBAS{\,a})$.\\
For \,$j = 1$, \,$[u]\mBAS{\,a}$ \,is odd and as 
\,$1\ \fleche{b/c}_{:_{a,2}}\ 0$ \,for \,$c \,= \,\frac{a+b}{2} \,< \,a$, 
\,$0\ \chemin{ub/vc}_{:_{a,2}}\ 0$ \,thus\\[0.25em]
\hspace*{6em}$a\,[u]\mBAS{\,a} + b \,= \,[ub]\mBAS{\,a} \,= 
\,2\,[vc]\mBAS{\,a}$ \ \ hence \ \
$[vc]\mBAS{\,a} \,= \,f'_{a,b}([u]\mBAS{\,a})$.
\cqfd
\end{proof}
Denoting by \,$c \,/_{\rm\!e}$ \,when \,$c$ \,is even, and by \,$c \,/_{\rm\!o}$
\,when \,$c$ \,is odd, we can describe the transducer \,${\cal G}'_{a,b}$ 
\,according to the parity of \,$a$ \,as follows:\\[0.5em]
\hspace*{2em}\includegraphics{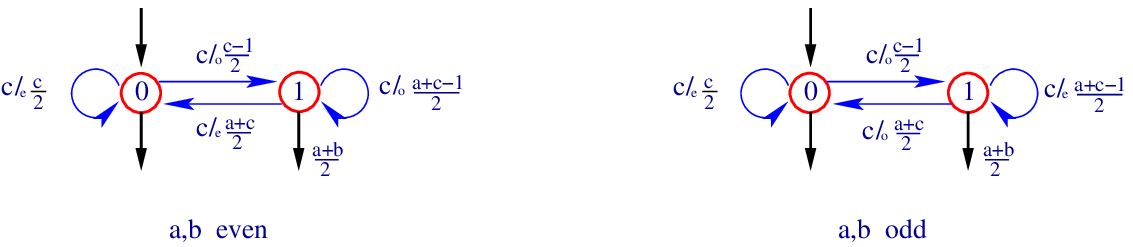}\\[0.5em]
So \,${\cal G}'_{3,1}$ \,realizes in base \,$3$ \,the acceleration \,$f'_{3,1}$ 
\,of the Collatz function\,:
\begin{center}
\includegraphics{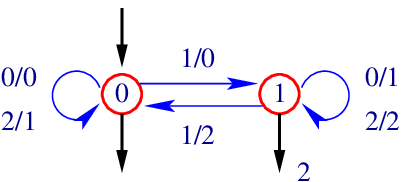}\\
\end{center}
Thus \,${\cal G}'_{6,2}$ \,realizes in base \,$6$ \,the Collatz function 
\,$f_{3,1} \,= \,f'_{6,2}$\,:
\begin{center}
\includegraphics{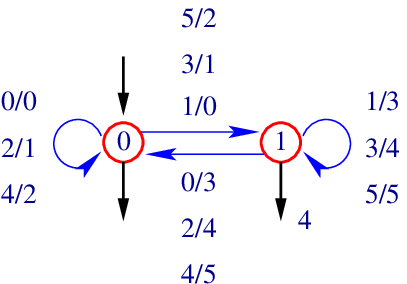}\\
\end{center}
More generally for any \,$0 \leq b < a$, the function 
\,$f_{a,b} \,= \,f'_{2a,2b}$ \,is represented in base \,$2a$ \,by the prefix 
sequential transducer \,${\cal G}'_{2a,2b}$\,. 
This synchronizability extends for any \,$f_{a,b,d}$\,. 
We define a prefix sequential transducer realizing \,$f_{a,b,d}$ \,for \,$b < a$ 
\,from the transducer computing the division by \,$d$ \,in base \,$ad$. 
If an initial path ends to the state \,$j$, the input represents an integer 
\,$n$ \,multiple of \,$d$ \,plus \,$j$ \,and the output represents 
\,$\frac{n-j}{d}$. The final digit in \,$j \neq 0$ \,is \,$aj+b$ \,since 
\,$ad\,\frac{n-j}{d} + aj + b \,= \,an + b \,= \,f_{a,b,d}(n)$.
\begin{theorem}\label{Transducf}
For all \,$0 \leq b < a \neq 1$ \,and \,$d > 0$, the following transducer 
\,${\cal G}_{a,b,d}$\\[0.25em]
\hspace*{0em}
$(\{0,\ldots,d-1\},\{0\},\omega_{a,b}\,,:_{ad,d})$ \,with 
\,$\omega_{a,b}(0) = \varepsilon$, $\omega_{a,b}(j) \,= \,aj + b$ 
\,$\forall \ 0 < j < d$\\[0.25em]
is prefix sequential and realizes a representation in base \,$ad$ \,of 
\,$f_{a,b,d}$\,.
\end{theorem}
The proof is similar to the one of Proposition~\ref{Transducf'}. 
Let us apply Theorem~\ref{Transducf}.
\begin{corollary}\label{SyncGauche}
For all integers \,$0 \leq b < a$ \,and \,$d > 0$, \,$f_{a,b,d}$ \,is prefix 
automatic.
\end{corollary}
The generalization of Corollary~\ref{SyncGauche} for all \,$b$ \,remains open.

\section{Prefix sequential transducers for \,${\bf f^{\,n}_{a,b,d}}$}

{\indent}For all natural numbers \,$n$ and for \,$b < a \neq 1$, we can make 
explicit the composition \,$n$ \,times \,${\cal G}^{\,n}_{a,b,d}$ \,of 
\,${\cal G}_{a,b,d}$\,.\\

First of all, we need to compose two Euclidean divisions in the same 
base~\,$a$\,: 
dividing by \,$d$ \,then by \,$d'$ \,corresponds to divide by \,$dd'$. 
Precisely, the relation \,$:_{a,d}\,\compose\,:_{a,d'}$ \,is in bijection
with the relation \,$:_{a,dd'}$ \,by coding any vertex \,$(i,i')$ \,where 
\,$0 \leq i < d$ \,and \,$0 \leq i' < d'$ \,by the integer 
\,$\mBAS{d}[(i,i')] \,= \,i + i'd$.
\begin{lemma}\label{CompDiv}
For all \,$a > 1$ \,and \,$d,d' > 0$, \ 
$\mBAS{d}[\,:_{a,d}\,\compose\,:_{a,d'}]$ \,is equal to \,$:_{a,dd'}$\,.
\end{lemma}
Thus and for any \,$n \geq 0$ \,and \,$x,y \in \{0,\ldots,d-1\}^n$, we get
\\[0.25em]
\hspace*{9em}$x\ \fleche{b/c}_{(:_{a,d})^n}\ y \ \ \ \Longleftrightarrow \ \ \ 
\mBAS{d}[x]\ \fleche{b/c}_{:_{a,d^n}}\ \mBAS{d}[y]$\\[0.25em]
and we can construct \,$(:_{a,d})^n$ \,from \,$:_{a,d^n}$ \,by renaming 
each integer vertex by the word of its reverse representation of length \,$n$ 
\,in base \,$d$.\\
To realize the function \,$f'^{\,n}_{a,b}$\,, \,we first do the division by 
\,$2^n$ \,and then the numerator is performed by the terminal function: it 
multiplies by \,$a$ \,and adds \,$b$ \,each time an odd number is encountered 
in the orbit of length \,$n$. 
We define for all natural numbers \,$a,b,n,q$ \,the number\\[0.25em]
\hspace*{8em}$\eta_{a,b,n}(q) \ = \ |\{\ 0 \leq i < n\ |\ f'^{\,i}_{a,b}(q)$ 
\,odd\ $\}|$\\[0.25em]
of odd integers among the first \,$n$ \,numbers of the orbit from \,$q$ \,of 
\,$f'_{a,b}$\,. 
Let us give a basic property \cite{Al} satisfied by the iterates of 
\,$f'_{a,b}$\,.
\begin{lemma}\label{formule}
For all natural numbers \,$a,b,n,p,q$ \,with \,$a,b$ \,of same parity,\\[0.25em]
\hspace*{9em}\begin{tabular}{rcl}
$f'^{\,n}_{a,b}(p2^n + q)$ & $\,=\,$ & $p\,a^{\eta_{a,b,n}(q)} + f'^{\,n}_{a,b}(q)$
\\[0.25em]
$\eta_{a,b,n}(p2^n + q)$ & $=$ & $\eta_{a,b,n}(q)$.
\end{tabular}\\
\end{lemma}
The prefix sequential transducer \,${\cal G}'^{\,n}_{a,b}$ \,realizing 
\,$f'^{\,n}_{a,b}$ \,is the synchronous sequential transducer of the division 
by \,$2^n$ \,in base \,$a$ \,with a terminal function defined by the \,$2^n$ 
\,first values of \,$f'^{\,n}_{a,b}$\,. 
\,The length of the final words is needed to take into account the number of 
\,$0$ \,to the left. 
For any vertex, the length of its final word is the number of odd numbers 
among the first \,$n$ \,values of its orbit.
\begin{proposition}\label{Transducf'Iter}
For all \,$n \geq 0$ \,and \,$0 \leq b < a \neq 1$ \,with \,$a,b$ 
\,of the same parity,\\[0.25em]
\hspace*{3em}$\mBAS{2}[{\cal G}'^{\,n}_{a,b}] \,= 
\,(\{0,\ldots,2^n\!-\!1\},\{0\},\omega'_n\,,:_{a,2^n})$ \ with for any 
\,$0 \leq i < 2^n$,\\[0.25em]
\hspace*{0.5em}$\omega'_n(i) \in \{0,\ldots,a-1\}^*$ \ such that 
\ $[\omega'_n(i)]\mBAS{\,a} \,= \,f'^{\,n}_{a,b}(i)$ \ and \ 
$|\omega'_n(i)| \,= \,\eta_{a,b,n}(i)$.
\end{proposition}
\begin{proof}
By induction on \,$n \geq 0$.\\
$n = 0$\,: \,${\cal G}'^{\,0}_{a,b} \,= 
\,(\{\varepsilon\},\{\varepsilon\},\omega,\{\,\varepsilon\ 
\fleche{c/c}\ \varepsilon \mid c \in \{0,\ldots,a-1\}\,\})$ \,with 
\,$\omega(\varepsilon) \,= \,\varepsilon$.\\[0.25em]
$n \ \Longrightarrow \ n+1$\,: \,we have 
\,${\cal G}'^{\,n+1}_{a,b} \,= \,{\cal G}'_{a,b}\,\compose\ {\cal G}'^{\,n}_{a,b}$
\,for the composition of prefix sequential transducers. 
By Lemma~\ref{CompDiv}, the relation 
\,$\mBAS{2}[\,:_{a,2}\,\compose\,:_{a,2^n}]$ \,is equal to \,$:_{a,2^{n+1}}$\,.\\
It remains to verify that \,$\omega'_{n+1}$ \,is the terminal function of 
\,$\mBAS{2}[{\cal G}'^{\,n+1}_{a,b}]$.\\
As \,$\omega'_{a,b}(0) \,= \,\varepsilon$, we get 
\,$\omega'_{n+1}(\mBAS{\,2}[0u]) \,= \,\omega'_n(\mBAS{\,2}[u])$ \,for any 
\,$u \in \{0,1\}^n$ \,{\it i.e.} 
\,$\omega'_{n+1}(2i) \,= \,\omega'_n(i)$ \,for all \,$0 \leq i < 2^n$. 
By induction hypothesis, we get\\[0.25em]
\hspace*{6em}\begin{tabular}{rllllll}
$[\omega'_{n+1}(2i)]\mBAS{\,a}$ & $\,=\,$ & $[\omega'_n(i)]\mBAS{\,a}$ & 
$\,=\,$ & $f'^{\,n}_{a,b}(i)$ & $\,=\,$ & $f'^{\,n+1}_{a,b}(2i)$\\[0.25em]
$|\omega'_{n+1}(2i)|$ & $\,=\,$ & $|\omega'_n(i)|$ & $\,=\,$ & 
$\eta_{a,b,n}(i)$ & $\,=\,$ & $\eta_{a,b,n+1}(2i)$.
\end{tabular}\\[0.25em]
Similarly \,$\omega'_{a,b}(1) \,= \,\frac{a+b}{2}$ \,and for any 
\,$0 \leq i < 2^n$, there exists unique \,$j$ \,and \,$c$ \,such that 
\,$i\,\ \flecheInd{\frac{a+b}{2}/c}{:_{a,2^n}}\ j$ 
\,thus \,$\omega'_{n+1}(2i+1) \,= \,c.\omega'_n(j)$.\\[0.25em]
Moreover \,$f'_{a,b}(2i+1) \,= \,ia + \frac{a+b}{2} \,= \,c2^n + j$. 
By Lemma~\ref{formule} and ind. hyp.,\\[0.25em]
\hspace*{1em}\begin{tabular}{rllllllllll}
$[\omega'_{n+1}(2i+1)]\mBAS{\,a}$ & $\,=\,$ & $[c.\omega'_n(j)]\mBAS{\,a}$ &
$\,=\,$ & $c\,a^{|\omega'_n(j)|} + [\omega'_n(j)]\mBAS{\,a}$ & $\,=\,$ & 
$c\,a^{\eta_{a,b,n}(j)} + f'^{\,n}_{a,b}(j)$\\[0.25em]
 & $\,=\,$ & $f'^{\,n}_{a,b}(c2^n + j)$ & $\,=\,$ & $f'^{\,n+1}_{a,b}(2i+1)$
\end{tabular}\\
and\\[0.25em]
\hspace*{1em}\begin{tabular}{rllllllllll}
$|\omega'_{n+1}(2i+1)|$ & $\,=\,$ & $1 + |\omega'_n(j)|$ & $\,=\,$ & 
$1 + \eta_{a,b,n}(j)$ & $\,=\,$ & $1 + \eta_{a,b,n}(c2^n + j)$\\[0.25em]
 & $\,=\,$ & $1 + \eta_{a,b,n}(f'_{a,b}(2i+1))$ & $\,=\,$ & $\eta_{a,b,n+1}(2i+1)$.
\end{tabular}
\cqfd
\end{proof}
Let us apply this proposition to the acceleration \,$f'_{5,1}$ \,of the 
\,$5n + 1$ \,function. 
We realize a representation in base $5$ of the third power 
\,$f'^{\,3}_{5,1}$ \,of \,$f'_{5,1}$ \,with the transducer 
\,${\cal G}'^{\,3}_{5,1}$ \,of transition set \,$(:_{5,2})^3$ \,of the division 
\,$:_{5,8}$ \,by \,$8$ \,in base \,$5$ \,which was illustrated in the previous 
section. Here is a presentation of its terminal function:
\begin{center}
\includegraphics{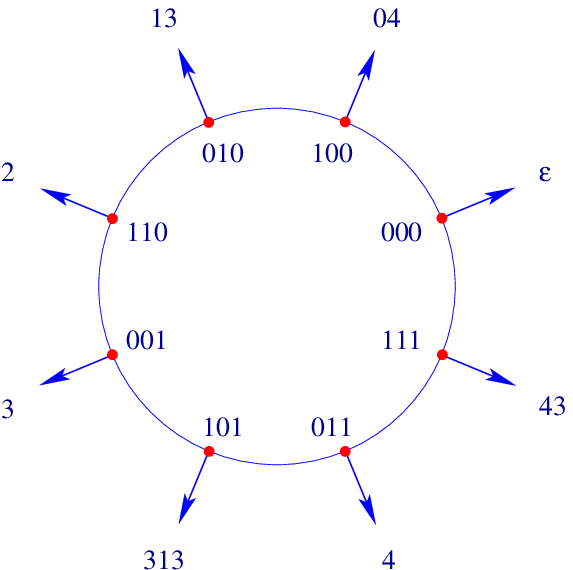}\\
\end{center}
Note that the final word of the state \,$100$ \,is \,$04$ \,and not \,$4$, 
and the accepting path \ 
$\fleche{}\,\ 000\ \,\fleche{4/0}\,\ 001\ \,\fleche{2/2}\,\ 011\,\ 
\fleche{3/4}\,\ 100\ \,\fleche{04}$ \ of input word \,$423$ \,and output word 
\,$02404$ \,represents 
\,$f'^{\,3}_{5,1}(113) \,= \,354$ \,in base \,$5$.\\
We can now give an explicit description of the prefix sequential transducer 
\,${\cal G}^{\,n}_{a,b,d}$ \,realizing \,$f^{\,n}_{a,b,d}$ \,for all \,$n$. 
For all \,$a,b,n,q \geq 0$, \,$\eta_{a,b,n}(q)$ \,is generalized to\\[0.25em]
\hspace*{3em}$\mu_{a,b,d,n}(q) \ = \ |\{\ 0 \leq i < n\ |\ f^{\,i}_{a,b,d}(q)$ 
\,not multiple of \,$d\ \}|$\\[0.25em]
the number of integers that are not multiples of \,$d$ \,among the first \,$n$ 
\,numbers of the orbit from \,$q$ \,of \,$f_{a,b,d}$\,. 
Let us adapt Lemma~\ref{formule} to the powers of \,$f_{a,b,d}$\,.
\begin{lemma}\label{formuleBis}
For all natural numbers \,$a,b,d,n,p,q$ \,with \,$d > 0$, we have\\[0.25em]
\hspace*{9em}\begin{tabular}{rcl}
$f_{a,b,d}^{\,n}(pd^n + q)$ & $\,=\,$ & 
$p\,(ad)^{\mu_{a,b,d,n}(q)} + f_{a,b,d}^{\,n}(q)$\\[0.25em]
$\mu_{a,b,d,n}(pd^n + q)$ & $=$ & $\mu_{a,b,d,n}(q)$.
\end{tabular}\\
\end{lemma}
Similarly to Proposition~\ref{Transducf'Iter}, we get an explicit description 
of \,${\cal G}^n_{a,b,d}$\,.
\begin{theorem}\label{TransducfIter}
For all integers \,$n \geq 0$ \,and \,$0 \leq b < a \neq 1$ \,and \,$d > 0$,
\\[0.25em]
\hspace*{3em}$\mBAS{d}[{\cal G}^{\,n}_{a,b,d}] \,= 
\,(\{0,\ldots,d^n\!-\!1\},\{0\},\omega_n\,,:_{ad,d^n})$ \ with for any 
\,$0 \leq i < d^n$,\\[0.25em]
\hspace*{1em}$\omega_n(i) \in \{0,\ldots,ad - 1\}^*$ \ with 
\ $[\omega_n(i)]\mBAS{\,ad} \,= \,f_{a,b,d}^n(i)$ \ and \ 
$|\omega_n(i)| \,= \,\mu_{a,b,d,n}(i)$.
\end{theorem}

\section{Prefix input-deterministic transducers for \,${\bf f^{\,*}_{a,b,d}}$}

{\indent}We present a simple infinite prefix input-deterministic transducer 
realizing the composition closure of \,$f_{a,b,d}$ \,for \,$b < a \neq 1$ 
\,and \,$d > 0$.\\

To give a description of a prefix input-deterministic transducer realizing the 
composition closure \,$f'^{\,*}_{a,b}$\,, \,we just take the union of the 
transition sets \,$(:_{a,2})^n$ \,of the division by \,$2^n$ \,in base \,$a$, 
plus the set of initial states \,$0^n$, plus a terminal function defined 
according to \,$b$ \,by length induction \,{\it i.e.} \,from the vertices of 
the division by \,$2^n$ \,to the vertices of the division by \,$2^{n-1}$.
\begin{proposition}\label{Resultat}
For all \,$0 \leq b < a$ \,with \,$a > 1$ \,and \,$a,b$ \,of the same parity,
\\[0.25em]
\hspace*{3em}
${\cal G}'^{\,*}_{a,b} \ = \ (\{0,1\}^*,0^*,\omega'_{a,b}\,,:_{a,2}^{\,*})$ 
\ where for all \,$u \in \{0,1\}^*$,\\[0.25em]
\hspace*{3em}$\omega'_{a,b}(0u) \,= \,\omega'_{a,b}(u)$ \ and \ 
$\omega'_{a,b}(1u) \,= \,c.\omega'_{a,b}(v)$ \ \ for \ 
$1u\ \fleche{b/c}_{:_{a,2}^{\,*}}\ 0v$.
\end{proposition}
\begin{proof}
We have seen in the proof of Proposition~\ref{Transducf'Iter} that for all 
\,$n \geq 0$, the terminal function \,$\omega'_{n+1}$ \,of 
\,${\cal G}'^{\,n+1}_{a,b}$ \,is defined recursively for all \,$0 \leq i < 2^n$ 
\,by\\[0.25em]
\hspace*{6em}\begin{tabular}{rcll}
$\omega'_{n+1}(2i)$ & $=$ & $\omega'_n(i)$\\[-0.25em]
$\omega'_{n+1}(2i+1)$ & $=$ & $c.\omega'_n(j)$ & for \ 
$i\,\ \flecheInd{\frac{a+b}{2}/c}{:_{a,2^n}}\ j$
\end{tabular}\\[0.25em]
and we have\\[0.25em]
\hspace*{1em}\begin{tabular}{rclcl}
$i\,\ \flecheInd{\frac{a+b}{2}/c}{:_{a,2^n}}\ j$ & $\,\Longleftrightarrow\,$ & 
$ia + \frac{a+b}{2} \,= \,c2^n + j$ & $\,\Longleftrightarrow\,$ & 
$(2i+1)a + b \,= \,c2^{n+1} + 2j$\\[0.5em]
 & $\Longleftrightarrow$ & $2i+1\ \fleche{b/c}_{:_{a,2^{n+1}}}\ 2j$.
\end{tabular}
\cqfd
\end{proof}
We visualize \,${\cal G}'^{\,*}_{a,b}$ \,by a cone with \,$\varepsilon$ \,at the 
tip and circular sections. The \,$n$-th section is the previously given 
representation of the Euclidean division \,$(:_{a,2})^n$ \,in base \,$a$ \,by 
\,$2^n$ \,of initial state \,$0^n$. The terminal function \,$\omega'_{a,b}$ \,is 
represented as follows:
\begin{center}
\includegraphics{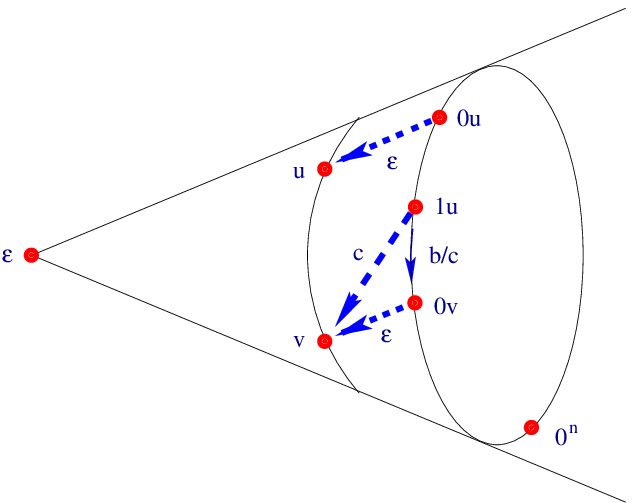}\\
\end{center}
with a transition \,$0u\ \fleche{\varepsilon}\ u$ \,from any node starting by 
\,$0$, and a transition \,$1u\ \fleche{c}\ v$ \,from any node starting by 
\,$1$ \,for the transition \,$1u\ \fleche{b/c}\ 0v$ \,of the division by 
\,$2^{|u|+1}$ in base \,$a$. 
Note that these transitions of the terminal function can only be used at the 
end of an accepting path.\\
Similarly to Proposition~\ref{Resultat}, we get an explicit description of the 
prefix input-deterministic transducer \,${\cal G}^{\,*}_{a,b,d}$\,.
\begin{theorem}\label{ResultatGene}
For all integers \,$0 \leq b < a \neq 1$ \,and \,$d > 0$,\\[0.25em]
\hspace*{1.5em}
${\cal G}^{\,*}_{a,b,d} \ = \ (\downarrow_d^*,0^*,\omega_{a,b,d}\,,\,:_{ad,d}^{\,*})$ 
\ with for all \,$u \in \,\downarrow_d^*$ \,and \,$0 < i < d$,\\[0.25em]
\hspace*{1.5em}$\omega_{a,b,d}(0u) \,= \,\omega_{a,b,d}(u)$ \ and \ 
$\omega_{a,b,d}(iu) \,= \,c.\omega_{a,b,d}(v)$ \ for \ 
$iu\ \fleche{bd/c}_{:_{ad,d}^{\,*}}\ 0v$.
\end{theorem}
Theorem~\ref{ResultatGene} states that under the condition \,$b < a \neq 1$, we 
realize the composition closure of \,$f_{a,b,d}$ \,by taking the union of the 
divisions \,$(:_{ad,d})^n$ \,by \,$d^n$ \,in base \,$ad$ \,of initial state 
\,$0^n$, plus a recurrent terminal function on \,$n$ \,dependent on \,$ab$.
\\[2em]
{\Large \bf Conclusion}\\[1em]
For any natural numbers \,$a,b,d$ \,with \,$b < a \neq 1$ \,and \,$d \neq 0$, 
we have given an explicit construction of an input-deterministic 
letter-to-letter transducer realizing the closure under composition of 
\,$f_{a,b,d}$. In its geometric representation, the disposition of the vertices 
is well appropriate for both the transitions of the Euclidean divisions and 
those of the terminal function. 
It might be a new approach to consider the circularity of the functions
\,$f_{a,b,d}$ \,namely the existence of paths \,$0^n\ \chemin{uv/0^{|v|}u}\ x$
\,where \,$v$ \,is the terminal word of the vertex \,$x$ \,in the transducer 
of the division by \,$d^n$ \,in base \,$ad$. 
However, the circularity of the Collatz function is already considered as a
difficult subproblem of the Collatz conjecture. 
With this approach and for the acceleration \,$f'_{3,1,2}$ \,of the Collatz
function, it comes down to a deeper understanding of the Euclidean divisions 
by \,$2^n$ \,in base \,$3$ \,and its transitions of input \,$1$.\\
In this paper, we studied the realizability of functions on integers by 
both sequential and synchronized transducers. 
It would be interesting to extend this paper to other functions.

\newpage\noindent
\pagenumbering{roman}
\setcounter{section}{0}
\renewcommand{\thesection}{\Alph{section}}
\section{Appendix}
We give here in details all the missing proofs.\\

\noindent{\large\bf 4 \ \ Prefix and suffix sequential transducers}\\

The two families of prefix sequential and of suffix sequential transducers 
are closed under composition.
\enonce{Proposition~\ref{PrefixSuffixSequential}}
{The prefix and the suffix sequential functions are preserved in quadratic time 
and space under composition, intersection and difference.}
\mbox{}\\[-1em]\mbox{}\noindent
{\bf i)} Let \,${\cal T} \,= \,(Q,i,\omega,T)$ \,and 
\,${\cal T}' \,= \,(Q',i',\omega',T')$ \,be prefix sequential transducers. 
We realize \,$\InfSup{\cal T}\,\compose\ \InfSup{{\cal T}'}$ \,by the 
composition of \,${\cal T}$ \,by \,${\cal T}'$ \,having this simpler form:
\\[0.25em]
\hspace*{1em}\begin{tabular}{lrcl}
 & ${\cal T}\,\compose\,{\cal T}'$ & $=$ & 
$(Q{\croix}Q',(i,i'),\omega\,\compose\,\omega',T\,\compose\,T')$\\[0.25em]
where & $T\,\compose\,T'$ & $=$ & 
$\{\ (p,p')\ \fleche{a/c}\ (q,q')\ |\ \exists\ b 
\ (p\ \fleche{a/b}_T\ q \,\wedge \,p'\ \fleche{b/c}_{T'}\ q')\ \}$\\[0.25em]
 & $(\omega\,\compose\,\omega')\,(q,p')$ & $=$ & $v.\omega'(q')$ \ \ 
for \ $p'\ \cheminInd{\,\omega(q)/v}{T'}\ q'$
\end{tabular}\\[0.25em]
that is illustrated as follows:\\[0.25em]
\hspace*{3em}{$\left.\begin{tabular}{l}
$\longrightarrow\ i\ \cheminInd{u/v}{T}\ q\ \fleche{w}_{\omega}$\\[0.25em]
$\longrightarrow\ i'\ \cheminInd{v/x}{T'}\ p'\ \cheminInd{w/y}{T'}\,
\fleche{z}_{\omega'}$
\end{tabular}\right\}
\ \longrightarrow\ (i,i')\ \chemin{u/x}_{_{T\,\compose\,T'}}\ (q,p')\ 
\fleche{yz}_{\omega\,\compose\,\omega'}$}\\[0.5em]
We realize \,$\InfSup{\cal T} \,\cap \,\InfSup{{\cal T}'}$ \,by the usual 
{\it synchronization product} \,of \,${\cal T}$ \,and \,${\cal T}'$\,:\\[0.25em]
\hspace*{1em}\begin{tabular}{lrcl}
 & ${\cal T}\croix{\cal T}'$ & $=$ & 
$(Q{\croix}Q',(i,i'),\omega{\croix}\omega',T{\croix}T')$\\[0.25em]
where & $T{\croix}T'$ & $=$ & $\{\ (p,p')\ \fleche{a/b}\ (q,q')\ |\ 
p\ \fleche{a/b}_T\ q \,\wedge \,p'\ \fleche{a/b}_{T'}\ q'\ \}$\\[0.25em]
and & $(\omega{\croix}\omega')(q,q')$ & $=$ & $\omega(q)$ \ when \ 
$\omega(q) \,= \,\omega'(q')$.
\end{tabular}\\[0.5em]
We also realize \,$\InfSup{\cal T} \,- \,\InfSup{{\cal T}'}$ \,by the following 
prefix sequential transducer:\\[0.25em]
\hspace*{1em}\begin{tabular}{lrcl}
 & ${\cal T}-{\cal T}'$ & $=$ & 
$(Q \cup (Q{\croix}Q'),(i,i'),\omega{\setminus}\omega',T{\setminus}T')$
\\[0.25em]
where & $T{\setminus}T'$ & $=$ & $(T{\croix}T') \,\cup \,T \,\cup 
\,\{\ (p,p')\ \fleche{a/b}\ q\ |\ p\ \fleche{a/b}_T\ q \,\wedge 
\,p'\ \nofleche{a/b}{}\,\!_{T'}\ \}$\\[0.25em]
and & $(\omega{\setminus}\omega')(q)$ & $=$ & $\omega(q)$ \ for any 
\,$q \in {\rm dom}(\omega)$\\[0.25em]
 & $(\omega{\setminus}\omega')(q,q')$ & $=$ & $\omega(q)$ \ for any 
\,$q \in {\rm dom}(\omega)$ \,and\\[0.25em]
 & & & \hspace*{6em}$(q' \notin {\rm dom}(\omega') \,\vee 
\,\omega'(q') \,\neq \,\omega(q))$.
\end{tabular}\\
The transducers \,${\cal T}\,\compose\,{\cal T}'$, ${\cal T}\croix{\cal T}'$, 
${\cal T}-{\cal T}'$ \,can be constructed from 
\,${\cal T}$ \,and \,${\cal T}'$ \,in quadratic time and space.\\[0.5em]
{\bf ii)} Let \,${\cal T} \,= \,(Q,i,F,T)$ \,and 
\,${\cal T}' \,= \,(Q',i',F',T')$ \,be suffix sequential transducers. 
The composition of \,${\cal T}$ \,by \,${\cal T}'$ \,has this simpler form:
\\[0.25em]
\hspace*{1em}\begin{tabular}{lrcl}
 & ${\cal T}\,\compose\,{\cal T}'$ & $=$ & 
$(Q{\croix}Q',(i,i'),F{\croix}F',T\,\compose\,T')$\\[0.25em]
where & $T\,\compose\,T'$ & $=$ & 
$\{\ (p,p')\ \fleche{a/\varepsilon}\ (q,p')\ |\ 
p\ \fleche{a/\varepsilon}_T\ q \,\wedge \,p' \in Q'\ \}$\\[0.25em]
 & & $\cup$ & $\{\ (p,p')\ \fleche{a/b}\ (q,q')\ |\ \exists\ c 
\ (p\ \fleche{a/c}_T\ q \,\wedge \,p'\ \fleche{c/b}_{T'}\ q')\ \}$.
\end{tabular}\\[0.25em]
The composition is illustrated as follows:\\[0.25em]
\hspace*{0em}{$\left.\begin{tabular}{l}
$\longrightarrow\ i\ \cheminInd{u/\varepsilon}{T}\,\cheminInd{v/x}{T}\ p\ 
\cheminInd{w/y}{T}\ q\ \longrightarrow$\\[0.25em]
$\longrightarrow\ i'\ \cheminInd{x/\varepsilon}{T'}\ p'\ 
\cheminInd{y/z}{T'}\ q'\ \longrightarrow$\\[0.25em]
$|v| = |x| \,\wedge \,|w| = |y| = |z|$
\end{tabular}\right\}
\ \longrightarrow\ (i,i')\ \chemin{uv/\varepsilon}_{_{T\,\compose\,T'}}\ (p,p')\ 
\chemin{w/z}_{_{T\,\compose\,T'}}\ (q,q')\ \longrightarrow$}\\[0.5em]
We have to show that \,$T\,\compose\,T'$ \,remains of initial 
\,$\varepsilon$-output.\\
Let \,$(p,p')\ \fleche{a/b}_{T \compose T'}\ (q,q')\ \fleche{c/d}_{T \compose T'}\ 
(r,r')$ \,with \,$b \neq \varepsilon$.\\
We have to check that \,$d \neq \varepsilon$.\\
There exists \,$e$ \,such that \,$p\ \fleche{a/e}_T\ q$ \,and 
\,$p'\ \fleche{e/b}_{T'}\ q'$.\\
As \,$T'$ \,is without \,$\varepsilon$-input transition, $e \neq \varepsilon$.\\
As \,$T$ \,is of initial \,$\varepsilon$-output, 
$q\ \nofleche{c/\varepsilon}{}\,\!_T\ r$. 
So \,$q'\ \fleche{\cdot/d}_{T'}\ r'$.\\
As \,$T'$ \,is of initial \,$\varepsilon$-output, $d \neq \varepsilon$.
\\[0.5em]
We realize \,$\InfSup{\cal T} \,\cap \,\InfSup{{\cal T}'}$ \,by the 
synchronization product of \,${\cal T}$ \,and \,${\cal T}'$\,:\\[0.25em]
\hspace*{1em}\begin{tabular}{lrcl}
 & ${\cal T}{\croix}{\cal T}'$ & $=$ & 
$(Q{\croix}Q',(i,i'),F{\croix}F',T{\croix}T')$\\[0.25em]
where & $T{\croix}T'$ & $=$ & 
$\{\ (p,p')\ \fleche{a/b}\ (q,q')\ |\ p\ \fleche{a/b}_T\ q \,\wedge 
\,p'\ \fleche{a/b}_{T'}\ q'\ \}$.
\end{tabular}\\[0.75em]
We realize \,$\InfSup{\cal T} \,- \,\InfSup{{\cal T}'}$ \,by the following 
suffix sequential transducer:\\[0.25em]
\hspace*{1em}\begin{tabular}{lrcl}
 & ${\cal T}\,{-}\,{\cal T}'$ & $=$ & 
$(Q \cup (Q{\croix}Q'),(i,i'),F \cup (F{\croix}(Q'-F')),T\,{\setminus}\,T')$.
\end{tabular}\\[0.5em]
These three transducers can be constructed in quadratic time and space 
accor\-ding to \,${\cal T}$ \,and \,${\cal T}'$. 
\cqfd\\

\noindent{\large\bf 6 \ \ Suffix automaticity of functions \,${\bf f_{a,b,d}}$}\\

Let us extend \,$_{d,a}\croix$ \,to paths.
\enonce{Lemma~\ref{PathMult}}
{For any \,$i,j \geq 0$ \,and \,$u,v \in \,\downarrow_d^*$\,, \,we have
\\[0.25em]
\hspace*{6em}$i\,\ \chemin{u/v}_{_{d,a}\croix}\ j \ \ \ \Longleftrightarrow \ \ \ 
\mBAS{d}[u]a + i \,= \,\mBAS{\,d}[v] + j\,d^{|u|}$ \,and \ $|u| = |v|$.}
\mbox{}\\[-1em]\mbox{}\noindent
$\Longrightarrow$\ : \,As \,$_{d,a}\croix$ \,is synchronous, \,$|u| = |v|$.\\ 
Let us check the equality by induction on \,$|u| \geq 0$.\\
$|u| = 0$\,: \,We have \,$u = \varepsilon = v$ \,and \,$i = j$ \,hence the 
equality.\\
Let \,$i\ \chemin{bu/cv}\ j$ \,with \,$b,c \in \,\downarrow_d$ \,and the 
implication true for \,$u$.\\
There exists \,$k$ \,such that \,$i\ \fleche{b/c}\ k\ \chemin{u/v}\ j$.\\
Thus \,$ab +i \,= \,c + dk$ \,and 
\,$\mBAS{d}[u]a + k \,= \,\mBAS{\,d}[v] + j\,d^{|u|}$. Hence\\[0.25em]
\hspace*{3em}\begin{tabular}{rclcl}
$\mBAS{\,d}[cv] + j\,d^{|bu|}$ & $\,=\,$ & 
$c + (\mBAS{\,d}[v] + j\,d^{|u|})d$ & $\,=\,$ & $c + (\mBAS{\,d}[u]a + k)d$
\\[0.25em]
 & $=$ & $\mBAS{d}[u]ad + c + dk$ & $=$ & $(\mBAS{\,d}[u]d + b)a + i$
\\[0.25em]
 & $=$ & $\mBAS{d}[bu]a + i$.
\end{tabular}\\[0.25em]
$\Longleftarrow$\ : \,by induction on \,$|u| \geq 0$.\\
$|u| = 0$\,: \,We have \,$u = \varepsilon = v$ \,and \,$i = j$ \,hence 
\,$i\ \chemin{u/v}\ j$.\\
Suppose the implication true for \,$|u|$ \,and 
\,$\mBAS{d}[bu]a + i \,= \,\mBAS{\,d}[cv] + j\,d^{|bu|}$ \,with \,$|u| = |v|$ 
\,and \,$0 \leq b,c < d$. So, we have\\[0.25em]
\hspace*{9em}
$(\mBAS{\,d}[u]a)d + ab + i \,= \,(\mBAS{\,d}[v] + j\,d^{|u|})d + c$.\\[0.25em]
By Euclidean division of \,$ab + i$ \,by \,$d$, we have 
\,$ab + i \,= \,kd + c'$ \,with \,$c' < d$.\\
As \,$c < d$, we have \,$c = c'$ \,hence 
\,$\mBAS{d}[u]a + k \,= \,\mBAS{\,d}[v] + j\,d^{|u|}$.\\
As \,$|u| \,= \,|v|$ \,and by induction hypothesis, \ $k\ \chemin{u/v}\ j$.\\
As \,$ab + i \,= \,kd + c$, we get \,$i\ \fleche{b/c}\ k$ \,hence 
\,$i\ \chemin{bu/cv}\ j$.
\cqfd\\

\noindent{\large\bf 7 \ \ Prefix automaticity of functions \,${\bf f_{a,b,d}}$}\\

We give a proof of the realization of \,$f_{a,b,d}$ \,by a prefix 
sequential transducer with \,$d$ \,states.
\enonce{Theorem~\ref{Transducf}}
{For all \,$0 \leq b < a \neq 1$ \,and \,$d > 0$, the following transducer 
\,${\cal G}_{a,b,d}$\\[0.25em]
\hspace*{0em}
$(\{0,\ldots,d-1\},\{0\},\omega_{a,b}\,,:_{ad,d})$ \,with 
\,$\omega_{a,b}(0) = \varepsilon$, $\omega_{a,b}(j) \,= \,aj + b$ 
\,$\forall \ 0 < j < d$\\[0.25em]
is prefix sequential and realizes a representation in base \,$ad$ \,of 
\,$f_{a,b,d}$\,.}
\mbox{}\\[-1em]\mbox{}\noindent
It suffices to generalize the proof of Proposition~\ref{Transducf'}.\\
As \,$:_{ad,d}$ \,is input-deterministic and input-complete, for any 
\,$u \in \,\downarrow_{ad}^*$\,, there exists a unique 
\,$v \in \,\downarrow_{ad}^*$ \,and a unique \,$j \in \,\downarrow_{d}$ \,such 
that \,$0\ \chemin{u/v}_{:_{ad,d}}\ j$.\\
By Lemma~\ref{CheminDiv}, 
\,$[u]\mBAS{\,ad} \,= \,d\,[v]\mBAS{\,ad} + j$.\\
For \,$j = 0$, \,$[u]\mBAS{\,ad}$ \,is a multiple of \,$d$ \,and 
\,$[v]\mBAS{\,ad} \,= \,\frac{[u]\mBAS{\,ad}}{d} \,= 
\,f_{a,b,d}([u]\mBAS{\,ad})$.\\
For \,$0 < j < d$, \,$[u]\mBAS{\,ad}$ \,is not multiple of \,$d$.\\
We have \,$j\ \ \flecheInd{(bd)/\omega_{a,b}(j)}{:_{ad,d}}\ 0$ \,hence 
\,$0\ \ \ \cheminInd{u(bd)/v.\omega_{a,b}(j)}{:_{ad,d}}\ \,0$ \,thus\\[0.25em]
\hspace*{6em}$ad\,[u]\mBAS{\,ad} + bd \,= \,[u(bd)]\mBAS{\,ad} \,= 
\,d\,[v.\omega_{a,b}(j)]\mBAS{\,ad}$\\
so\\
\hspace*{6em}$[v.\omega_{a,b}(j)]\mBAS{\,ad} \,= \,a\,[u]\mBAS{\,ad} + b \,= 
\,f_{a,b,d}([u]\mBAS{\,ad})$.
\cqfd\\
It follows the prefix automaticity of \,$f_{a,b,d}$\,.
\enonce{Corollary~\ref{SyncGauche}}
{For all integers \,$0 \leq b < a$ \,and \,$d > 0$, \,$f_{a,b,d}$ \,is prefix 
automatic.}\mbox{}\\[-1em]\mbox{}\noindent
It remains to check the case \,$a = 1$ \,and \,$b = 0$. 
For \,$d > 0$, the transducer\\[0.25em]
\hspace*{1em}$(\downarrow_d\,,0,\omega,\{\,i\ \fleche{j/i}\ j \mid 
i,j \in \,\downarrow_d\,\})$ \,with \,$\omega(0) = \varepsilon$ \,and 
\,$\omega(i) = i$ for any \,$0 < i <d$\\[0.25em]
is prefix sequential and realizes\\[0.25em]
\hspace*{6em}$\{\,(u0,0u) \mid u \in \,\downarrow_d^*\,\} \,\cup 
\,\{\,(ui,0ui) \mid u \in \,\downarrow_d^* \,\wedge \ 0 < i < d\,\}$\\[0.25em]
which is a representation of \,$f_{1,0,d}$ \,in base \,$d$ \,(with the least 
significant digit to the right).
\cqfd\\

\noindent{\large\bf 8 \ \ Prefix sequential transducers for 
\,${\bf f^{\,n}_{a,b,d}}$}\\

Let us express the composition of two divisions in the same base.
\enonce{Lemma~\ref{CompDiv}}
{For all \,$a > 1$ \,and \,$d,d' > 0$, \ 
$\mBAS{d}[\,:_{a,d}\,\compose\,:_{a,d'}]$ \,is equal to \,$:_{a,dd'}$\,.}
\mbox{}\\[-1em]\mbox{}\noindent
For all \,$0 \leq i,j < d$ \,and \,\,$0 \leq i',j' < d'$, we have\\[0.25em]
\hspace*{3em}\begin{tabular}{ll}
 & $(i,i')\ \fleche{b/c}_{:_{a,d}\,\compose\,:_{a,d'}}\ (j,j')$\\[0.25em]
$\Longleftrightarrow\ $ & 
$\exists \ 0 \leq e < a$ \,such that \,$i\ \fleche{b/e}_{:_{a,d}}\ j$ \,and 
\,$i'\ \chemin{e/c}_{:_{a,d'}}\ j'$\\[0.25em]
$\Longleftrightarrow$ & $\exists \ 0 \leq e < a$ \,such that 
\,$ia + b \,= \,ed + j$ \,and \,$i'a + e \,= \,cd' + j'$\\[0.25em]
$\Longleftrightarrow$ & $ia + b \,= \,(cd' + j' - i'a)d + j$\\[0.25em]
$\Longleftrightarrow$ & $(i + i'd)a + b \,= \,cdd' + j + j'd$\\[0.25em]
$\Longleftrightarrow$ & 
$\mBAS{d}[(i,i')]\ \fleche{b/c}_{:_{a,dd'}}\ \mBAS{d}[(j,j')]$\,.
\end{tabular}
\cqfd\\
To be complete, we verify a known property on the iterates of \,$f'_{a,b}$\,.
\enonce{Lemma~\ref{formule}}
{For all natural numbers \,$a,b,n,p,q$ \,with \,$a,b$ \,of same parity,
\\[0.25em]
\hspace*{9em}\begin{tabular}{rcl}
$f'^{\,n}_{a,b}(p2^n + q)$ & $\,=\,$ & $p\,a^{\eta_{a,b,n}(q)} + f'^{\,n}_{a,b}(q)$
\\[0.25em]
$\eta_{a,b,n}(p2^n + q)$ & $=$ & $\eta_{a,b,n}(q)$.
\end{tabular}}\mbox{}\\[-1em]\mbox{}\noindent
By induction on \,$n \geq 0$.\\
$n = 0$\,: \,$\eta_{a,b,0}$ \,is the constant mapping \,$0$ \,and 
\,$f'^{\,0}_{a,b}$ \,is the identity.\\
$n \,\Longrightarrow \,n+1$\,: \,For \,$q$ \,even, we have\\[0.25em]
\hspace*{1em}\begin{tabular}{rclcl}
$f'^{\,n+1}_{a,b}(p2^{n+1} + q)$ & $=$ & $f'^{\,n}_{a,b}(f'_{a,b}(p2^{n+1} + q))$ & 
$=$ & $f'^{\,n}_{a,b}(p2^n + \frac{q}{2})$\\[0.25em]
 & $=$ & $p\,a^{\eta_{a,b,n}(\frac{q}{2})} + f'^{\,n}_{a,b}(\frac{q}{2})$ & $=$ & 
$p\,a^{\eta_{a,b,n+1}(q)} + f'^{\,n+1}_{a,b}(q)$
\end{tabular}\\[0.25em]
and \ \ $\eta_{a,b,n+1}(p2^{n+1} + q) \ = \ \eta_{a,b,n}(p2^n + \frac{q}{2}) \ = 
\ \eta_{a,b,n}(\frac{q}{2}) \ = \ \eta_{a,b,n+1}(q)$.\\[0.25em]
For \,$q$ \,odd, we have\\[0.25em]
\hspace*{0em}\begin{tabular}{rclcl}
$f'^{\,n+1}_{a,b}(p2^{n+1} + q)$ & $=$ & $f'^{\,n}_{a,b}(f'_{a,b}(p2^{n+1} + q))$ & 
$=$ & $f'^{\,n}_{a,b}(ap2^n + \frac{aq + b}{2})$\\[0.25em]
 & $=$ & $f'^{\,n}_{a,b}(ap2^n + f'_{a,b}(q))$ & $=$ & 
$p\,a^{1+\eta_{a,b,n}(f'_{a,b}(q))} + f'^{\,n}_{a,b}(f'_{a,b}(q))$\\[0.25em]
  & & & $=$ & $p\,a^{\eta_{a,b,n+1}(q)} + f'^{\,n+1}_{a,b}(q)$\\[0.5em] 
and \ \ \ \ \ \ \ \ \ \ \ \ \ \ \ \ \ \ \ \ & \\[0.25em]
$\eta_{a,b,n+1}(p2^{n+1} + q)$ & $=$ &
$1 + \eta_{a,b,n}(f'_{a,b}(p2^{n+1} + q))$ & $=$ & 
$1 + \eta_{a,b,n}(ap2^n + f'_{a,b}(q))$\\[0.25em]
 & $=$ & $1 + \eta_{a,b,n}(f'_{a,b}(q))$ & $=$ & $\eta_{a,b,n+1}(q)$.
\end{tabular}
\cqfd\\
Similarly to Lemma~\ref{formule}, we give a basic property satisfied by the 
powers of \,$f_{a,b,d}$\,.
\enonce{Lemma~\ref{formuleBis}}
{For all natural numbers \,$a,b,d,n,p,q$ \,with \,$d > 0$, we have\\[0.25em]
\hspace*{9em}\begin{tabular}{rcl}
$f_{a,b,d}^{\,n}(pd^n + q)$ & $\,=\,$ & 
$p\,(ad)^{\mu_{a,b,d,n}(q)} + f_{a,b,d}^{\,n}(q)$\\[0.25em]
$\mu_{a,b,d,n}(pd^n + q)$ & $=$ & $\mu_{a,b,d,n}(q)$.
\end{tabular}}\mbox{}\\[-1em]\mbox{}\noindent
By induction on \,$n \geq 0$.\\
$n = 0$\,: \,immediate because \,$\mu_{a,b,d,0}$ \,is the constant mapping \,$0$ 
\,and \,$f^0_{a,b,d}$ \,is the identity.\\
$n \,\Longrightarrow \,n+1$\,: \,For \,$q$ \,multiple of \,$d$, we have
\\[0.25em]
\hspace*{1em}\begin{tabular}{rclcl}
$f^{\,n+1}_{a,b,d}(pd^{n+1} + q)$ & $\,=\,$ & 
$f^{\,n}_{a,b,d}(f_{a,b,d}(pd^{n+1} + q))$ & $\,=\,$ & 
$f^{\,n}_{a,b,d}(pd^n + \frac{q}{d})$\\[0.25em]
 & $=$ & $p\,(ad)^{\mu_{a,b,d,n}(\frac{q}{d})} + f^{\,n}_{a,b,d}(\frac{q}{d})$ & $=$ & 
$p\,(ad)^{\mu_{a,b,d,n+1}(q)} + f^{\,n+1}_{a,b,d}(q)$
\end{tabular}\\[0.25em]
and \ \ $\mu_{a,b,d,n+1}(pd^{n+1} + q) \ = \ \mu_{a,b,d,n}(pd^n + \frac{q}{d}) \ = 
\ \mu_{a,b,d,n}(\frac{q}{d}) \ = \ \mu_{a,b,d,n+1}(q)$.\\[0.25em]
For \,$q$ \,not multiple of \,$d$, we have\\[0.25em]
\hspace*{3em}\begin{tabular}{clcl}
  & $f^{\,n+1}_{a,b,d}(pd^{n+1} + q)$\\[0.25em]
$=$ & $f^{\,n}_{a,b,d}(f_{a,b,d}(pd^{n+1} + q))$ & 
$=$ & $f^{\,n}_{a,b,d}(apd^{n+1} + aq + b)$\\[0.25em]
$=$ & $f^{\,n}_{a,b,d}((pad)d^n + f_{a,b,d}(q))$ & $=$ & 
$pad\,(ad)^{\mu_{a,b,d,n}(f_{a,b,d}(q))} + f^{\,n}_{a,b,d}(f_{a,b,d}(q))$\\[0.25em]
$=$ & $p\,(ad)^{\mu_{a,b,d,n+1}(q)} + f^{\,n+1}_{a,b,d}(q)$\\[0.5em] 
and & \\[0.25em]
 & $\mu_{a,b,d,n+1}(pd^{n+1} + q)$\\[0.25em]
$=$ & $1 + \mu_{a,b,d,n}(apd^{n+1} + aq + b)$ & $=$ & 
$1 + \mu_{a,b,d,n}((pad)d^n + f_{a,b,d}(q))$\\[0.25em]
$=$ & $1 + \mu_{a,b,d,n}(f_{a,b,d}(q))$ & $=$ & $\mu_{a,b,d,n+1}(q)$.
\end{tabular}
\cqfd\\
We explicit the composition \,$n$ \,times of the transducer 
\,${\cal G}_{a,b,d}$\,.
\enonce{Theorem~\ref{TransducfIter}}
{For all integers \,$n \geq 0$ \,and \,$0 \leq b < a \neq 1$ \,and \,$d > 0$,
\\[0.25em]
\hspace*{3em}$\mBAS{d}[{\cal G}^{\,n}_{a,b,d}] \,= 
\,(\{0,\ldots,d^n\!-\!1\},\{0\},\omega_n\,,:_{ad,d^n})$ \ with for any 
\,$0 \leq i < d^n$,\\[0.25em]
\hspace*{1em}$\omega_n(i) \in \{0,\ldots,ad - 1\}^*$ \ with 
\ $[\omega_n(i)]\mBAS{\,ad} \,= \,f_{a,b,d}^n(i)$ \ and \ 
$|\omega_n(i)| \,= \,\mu_{a,b,d,n}(i)$.}\mbox{}\\[-1em]\mbox{}\noindent
By induction on \,$n \geq 0$.\\
$n = 0$\,: \,${\cal G}^{\,0}_{a,b,d} \,= 
\,(\{\varepsilon\},\{\varepsilon\},\omega,\{\,\varepsilon\ 
\fleche{c/c}\ \varepsilon \mid c \in \{0,\ldots,ad-1\}\,\})$ \,with 
\,$\omega(\varepsilon) \,= \,\varepsilon$.\\[0.25em]
$n \ \Longrightarrow \ n+1$\,: \,we have 
\,${\cal G}^{\,n+1}_{a,b,d} \,= \,{\cal G}_{a,b,d}\,\compose\ {\cal G}^{\,n}_{a,b,d}$
\,for the composition of prefix sequential transducers.\\
By Lemma~\ref{CompDiv}, the transition relation 
\,$\mBAS{d}[\,:_{ad,d}\,\compose\,:_{ad,d^n}]$ \,is equal to \,$:_{ad,d^{n+1}}$\,.\\
It remains to verify that \,$\omega_{n+1}$ \,is the terminal function of 
\,$\mBAS{d}[{\cal G}^{\,n+1}_{a,b,d}]$.\\
As \,$\omega_{a,b}(0) \,= \,\varepsilon$, we get 
\,$\omega_{n+1}(\mBAS{\,d}[0u]) \,= \,\omega_n(\mBAS{\,d}[u])$ \,for any 
\,$u \in \downarrow^n_d$\\
{\it i.e.} \,$\omega_{n+1}(di) \,= \,\omega_n(i)$ \,for all \,$0 \leq i < d^n$. 
By induction hypothesis, we get\\[0.25em]
\hspace*{6em}\begin{tabular}{rllllll}
$[\omega_{n+1}(di)]\mBAS{\,ad}$ & $\,=\,$ & $[\omega_n(i)]\mBAS{\,ad}$ & 
$\,=\,$ & $f^{\,n}_{a,b,d}(i)$ & $\,=\,$ & $f^{\,n+1}_{a,b,d}(di)$\\[0.25em]
$|\omega_{n+1}(di)|$ & $\,=\,$ & $|\omega_n(i)|$ & $\,=\,$ & 
$\mu_{a,b,d,n}(i)$ & $\,=\,$ & $\mu_{a,b,d,n+1}(di)$.
\end{tabular}\\[0.25em]
Let \,$0 \leq i < d^n$ \,and \,$0 < j < d$. 
So \,$\omega_{a,b}(j) \,= \,aj + b \,\leq \,a(d-1) + b \,< \,ad$.\\
There exists unique \,$k$ \,and \,$c$ \,such that 
\,$i\,\ \flecheInd{aj+b/c}{:_{ad,d^n}}\ k$\\
thus \ $\omega_{n+1}(di+j) \,= \,c.\omega_n(k)$.\\[0.25em]
Moreover \,$f_{a,b,d}(di+j) \,= \,iad + aj + b \,= \,cd^n + k$. 
By Lemma~\ref{formuleBis} and ind. hyp.,\\[0.25em]
\hspace*{0em}\begin{tabular}{rllllllllll}
$[\omega_{n+1}(di+j)]\mBAS{\,ad}$ & $\,=\,$ & $[c.\omega_n(k)]\mBAS{\,ad}$ &
$\,=\,$ & $c\,(ad)^{|\omega_n(k)|} + [\omega_n(k)]\mBAS{\,ad}$ & $\,=\,$ & 
$c\,(ad)^{\mu_{a,b,d,n}(k)} + f^{\,n}_{a,b,d}(k)$\\[0.25em]
 & $\,=\,$ & $f^{\,n}_{a,b,d}(cd^n + k)$ & $\,=\,$ & $f^{\,n+1}_{a,b,d}(di+j)$
\end{tabular}\\
and\\[0.25em]
\hspace*{0em}\begin{tabular}{rllllllllll}
$|\omega_{n+1}(di+j)|$ & $\,=\,$ & $1 + |\omega_n(k)|$ & $\,=\,$ & 
$1 + \mu_{a,b,d,n}(k)$ & $\,=\,$ & $1 + \mu_{a,b,d,n}(cd^n + k)$\\[0.25em]
 & $\,=\,$ & $1 + \mu_{a,b,d,n}(f_{a,b,d}(di+j))$ & $\,=\,$ & 
$\mu_{a,b,d,n+1}(di+j)$.
\end{tabular}
\cqfd\\

\noindent{\large\bf 9 \ \ Prefix input-deterministic transducers for 
\,${\bf f^{\,*}_{a,b,d}}$}\\

We explicit the closure under composition of the transducer 
\,${\cal G}_{a,b,d}$\,.
\enonce{Theorem~\ref{ResultatGene}}
{For all integers \,$0 \leq b < a \neq 1$ \,and \,$d > 0$,\\[0.25em]
\hspace*{1.5em}
${\cal G}^{\,*}_{a,b,d} \ = \ (\downarrow_d^*,0^*,\omega_{a,b,d},\,:_{ad,d}^{\,*})$ 
\ with for all \,$u \in \,\downarrow_d^*$ \,and \,$0 < i < d$,\\[0.25em]
\hspace*{1.5em}$\omega_{a,b,d}(0u) \,= \,\omega_{a,b,d}(u)$ \ and \ 
$\omega_{a,b,d}(iu) \,= \,c.\omega_{a,b,d}(v)$ \ for \ 
$iu\ \fleche{bd/c}_{:_{ad,d}^{\,*}}\ 0v$.}\mbox{}\\[-1em]\mbox{}\noindent
We have seen in the proof of Theorem~\ref{TransducfIter} that for all 
\,$n \geq 0$, the terminal function \,$\omega_{n+1}$ \,of 
\,${\cal G}^{\,n+1}_{a,b,d}$ \,is defined recursively for all \,$0 \leq i < d^n$ 
\,and \,$0 < j < d$~\,by\\[0.25em]
\hspace*{6em}\begin{tabular}{rcll}
$\omega_{n+1}(di)$ & $\,=\,$ & $\omega_n(i)$\\[-0.25em]
$\omega_{n+1}(di + j)$ & $\,=\,$ & $c.\omega_n(k)$ & \ for \ 
$i\,\ \fleche{aj + b/c}_{:_{ad,d^n}}\ k$
\end{tabular}\\[0.25em]
and we have\\[0.25em]
\hspace*{0em}\begin{tabular}{rclcl}
$i\,\ \fleche{aj + b/c}_{:_{ad,d^n}}\ k$ & $\,\Longleftrightarrow\,$ & 
$iad + aj + b \,= \,cd^n + k$ & $\,\Longleftrightarrow\,$ & 
$(di+j)ad + bd \,= \,cd^{n+1} + dk$\\[0.5em]
 & $\Longleftrightarrow$ & $di + j\,\ \fleche{bd/c}_{:_{ad,d^{n+1}}}\ dk$\,.
\end{tabular}
\cqfd

\begin{thebibliography}{8}

\bibitem{Al} J.-P. Allouche\,: 
\textsl{T. Tao et la conjecture de Syracuse}, 
Gazette de la SMF~168, 34--39 (2021).

\bibitem{Be} J. Berstel\,: 
\textsl{Transductions and context-free languages}, 
Teubner Studienb\"ucher Informatik (1979).

\bibitem{Ei} S. Eilenberg\,: 
\textsl{Automata, languages and machines}, Vol. A, Academic Press (1974).

\bibitem{EM} C. Elgot and J. Mezei\,: 
\textsl{On relations defined by generalized finite automata},
IBM journal of research and development~9, 47--68 (1965)
  
\bibitem{FS} C. Frougny and J. Sakarovitch\,: 
\textsl{Synchronized rational relations of finite and infinite words},
Theoretical Computer Science~108~(1), 45--82 (1993).

\bibitem{Gi} S. Ginsburg\,: 
\textsl{The mathematical theory of context-free languages}, 
McGraw-Hill (1966).

\bibitem{La} J. Lagarias\,: 
\textsl{The ultimate challenge: the 3x+1 problem}, 
American Mathematical Society (2010).

\bibitem{Me} G. Mealy\,: 
\textsl{A Method for Synthesizing Sequential Circuits}, 
Bell System Technical Journal, 1045--1079 (1955).

\bibitem{RS} M. Rabin and D. Scott\,: 
\textsl{Finite automata and their decision problems},
IBM journal of research and development~3(2),
125--144 (1959). 
Reprinted in \textsl{Sequential machines: Selected papers}, 
Addison Wesley (1965).

\bibitem{Sa} J. Sakarovitch\,: 
\textsl{Elements of automata theory}, 
Cambridge University Press (2009).
  
\bibitem{Sc} M.-P. Sch\"utzenberger\,: 
\textsl{Sur une variante des fonctions s\'equentielles}, 
Theoretical Computer Science~4~(1), 47--57 (1977).

\end{thebibliography}
\end{document}